\documentclass[runningheads,orivec]{llncs}

\usepackage{amsmath}
\usepackage{amssymb}
\usepackage{amsthm}
\usepackage{booktabs}
\usepackage[anycase]{thmtools}

\usepackage{cite}

\usepackage{array}
\usepackage{multirow}

\usepackage{algorithm}
\usepackage[noend]{algpseudocode}

\usepackage{tikz}
\usetikzlibrary{arrows,automata,calc,patterns}
\tikzset{every picture/.append style={->, >=stealth', shorten >=0pt, auto, initial text={}, initial above}}
\tikzstyle{Takenode}=[circle,draw,inner sep=0pt,minimum size=25pt,pattern color={red!25!white},pattern={north east lines}]
\tikzstyle{Notakenode}=[circle,draw,inner sep=0pt,minimum size=25pt,pattern color={blue!25!white},pattern={crosshatch dots}]
\tikzstyle{fixedsize}=[text width=1.2cm, minimum height=1cm, text=black, scale=1, align=center]

\usepackage{hyperref}

\usepackage{orcidlink}
\renewcommand{\orcidID}[1]{\orcidlink{#1}}
\newcommand{\Letter}{{\fontfamily{mvs}\fontencoding{U}\fontseries{m}\fontshape{n}\selectfont\char66}}

\newcommand{\functionName}[1]{\textsc{#1}}

\renewcommand{\epsilon}{\varepsilon}
\renewcommand{\phi}{\varphi}

\newcommand{\naturals}{\mathbb{N}}

\newcommand{\reals}{\mathbb{R}}
\newcommand{\posreals}{\reals_{\geq 0}}

\newcommand{\indvec}[1]{\mathbf{1}_{#1}}
\newcommand{\transpose}[1]{#1^{\intercal}}

\newcommand{\abs}[1]{\lvert #1 \rvert}

\newcommand{\setnocond}[1]{\{#1\}}
\newcommand{\setcond}[2]{\{\, #1 \mid #2 \, \}}

\newcommand{\mystackrelsingle}[2]{%
  \mathrel{\vbox{\offinterlineskip\ialign{%
    \hfil##\hfil\cr%
    $\scriptstyle#1$\cr%
    \noalign{\kern.3ex}%
    $#2$\cr%
}}}}

\newcommand{\modelsymbol}[1]{\mathcal{#1}}
\newcommand{\mc}[1][D]{\modelsymbol{#1}}

\newcommand{\mstates}{S}
\newcommand{\minit}[1][s]{\bar{#1}}

\newcommand{\mtransitions}{\mathbf{P}}

\newcommand{\initdistribution}[1][]{\minit[\iota]#1}

\newcommand{\mpath}{\sigma}

\newcommand{\evaluation}{\nu}

\newcommand{\MooreInput}{\Sigma}
\newcommand{\MooreOutput}{\Gamma}
\newcommand{\MooreTransition}{T}
\newcommand{\Outfunc}{O}
\newcommand{\kbit}{k-bit }
\newcommand{\taken}{\mathrm{T}}
\newcommand{\notaken}{\mathrm{NT}}
\newcommand{\strongtaken}{\mathrm{ST}}
\newcommand{\weaktaken}{\mathrm{WT}}
\newcommand{\strongnotaken}{\mathrm{SN}}
\newcommand{\weaknotaken}{\mathrm{WN}}

\newcommand{\pnbit}{probabilistic k-bit }

\newcommand{\attackout}[1][v]{\mathsf{out}_{#1}}
\newcommand{\ranvare}[1]{z_{#1}}
\newcommand{\ranvars}[1]{\mathcal{Z}_{#1}}
\newcommand{\safepstaken}{\strongtaken}
\newcommand{\safepsnotaken}{\strongnotaken}

\newcommand{\privacypnbit}{enhanced k-bit probabilistic}
\newcommand{\threshold}{m}
\newcommand{\privacy}{p}
\newcommand{\execute}{\mathsf{exec}}
\newcommand{\pstationary}{\mu}
\newcommand{\tranmatrix}{\mathrm{M}}
\newcommand{\substituteg}{g}

\begin{document}

\title{Synthesizing Probabilistic Saturating Counters with Differentially Private Formal Guarantees}

\author{
     Zhiming Chi\inst{1,2,3}\orcidID{0009-0003-1110-5203}
     \and
     Lutan Zhao\inst{3,4}\orcidID{0000-0003-3672-6463}
     \and
     Depeng Liu\inst{1,2}\orcidID{0000-0002-9353-9691}
     \and
     Yong Li\inst{1,2}\orcidID{0000-0002-7301-9234}
     \and
     \\
     Pengfei Yang\inst{5}\orcidID{0000-0003-4114-7757}
     \and
     Bow-Yaw Wang\inst{6}\orcidID{0000-0002-5757-545X}
     \and
     Rui Hou\inst{3,4}\orcidID{0000-0002-9215-7632}
     \and
     Cheng-Chao Huang\inst{7}\orcidID{0000-0002-9693-8778}
     \and
     \\
     Andrea Turrini\inst{1,2}\orcidID{0000-0003-4343-9323}
     \and
     Lijun Zhang\inst{1,2,3}\textsuperscript{(\Letter)}\orcidID{0000-0002-3692-2088}
    \and
    Naijun Zhan\inst{8}\orcidID{0000-0003-3298-3817}
 }
\institute{
    Key Laboratory of System Software (Chinese Academy of Sciences), Beijing, China
    \and
    Institute of Software, Chinese Academy of Sciences, Beijing, China
    \and
    University of Chinese Academy of Sciences, Beijing, China
    \and
    Institute of Information Engineering, Chinese Academy of Sciences, Beijing, China
    \and
    College of Computer and Information Science, Software College, Southwest University, Chongqing, China
    \and
    Institute of Information Science, Academia Sinica, Taipei, Taiwan
    \and
    Nanjing Institute of Software Technology, Chinese Academy of Sciences, Nanjing, China
    \and
    School of Computer Science \& MOE Key Laboratory of High Confidence Software Technologies, Peking University, Beijing, China 
    \\
    \email{zhanglj@ios.ac.cn}
}

\maketitle

\setcounter{footnote}{0}

\pagestyle{plain}

\begin{abstract}
Branch predictors improve instruction-level parallelism in modern processors and are commonly modeled using saturating counters.
However, classical saturating counters are deterministic and thus vulnerable to side-channel attacks: an attacker can manipulate the counter state and infer the branch direction of a victim process.
Probabilistic saturating counters (PSCs) have been proposed to mitigate this leakage by randomizing counter updates, but existing evaluations are mainly empirical.
In this paper, we give a formal analysis based on differential privacy (DP):
we model PSCs and the corresponding Prime+Probe attack strategies as probabilistic Moore machines, derive optimal attack strategies, and quantify the attacker's distinguishing power through DP.
Our DP guarantee applies to the PSC primitive under the Prime+Probe observation model; end-to-end security for a full branch predictor under repeated or adaptive attacks is an important direction for future work.
We then synthesize parameters for an enhanced PSC that satisfies a target pure DP guarantee.
To evaluate utility, we derive the stationary misprediction rate and validate the theoretical predictions on benchmark programs.
Compared to deterministic and existing probabilistic saturating counters, the synthesized PSCs provide formal security guarantees while preserving competitive prediction performance.
\end{abstract}

\section{Introduction}

Branch prediction is a fundamental technique for enhancing instruction-level parallelism in modern high-performance processors~\cite{SamsungExynosISCA2020,IBMPower8,AMDZen2Paper}. 
However, the shared nature of the predictor's resources, whether in single-core or Simultaneous Multi-Threading processors, has exposed significant security vulnerabilities~\cite{evtyushkin2018branchscope,lee2017inferring,huo2020bluethunder}. 
The root cause lies in the deterministic update strategy of the Saturating Counter (SC)--the fundamental building block of predictors ranging from GShare to TAGE~\cite{mcfarling1993combining,seznec_TAGE_SC_L}. 
This determinism allows attackers to manipulate the counter state and exploit side-channels, such as the Prime+Probe attack~\cite{JCST2021zhao,evtyushkin2018branchscope}, to infer fine-grained execution traces and thus extract information about the victim's private data. 
While Probabilistic Saturating Counters (PSCs)~\cite{JCST2021zhao} have been proposed to mitigate these attacks by introducing randomized transitions, existing security analyses rely on incomplete simulations, failing to rigorously balance the trade-off between security and performance.

In this paper, we address this gap by applying formal verification to the design and analysis of PSCs, specifically targeting the Prime+Probe attack. 
In this attack scenario, an adversary first primes the target saturating counter to a known initial state and subsequently probes the counter to detect the cut-off point--the first instance of a correct prediction during the victim's execution. 
By observing this cut-off point, the adversary can infer the victim's branch direction; 
repeated attacks allow the attacker to infer the private data used in the branch decision. 
We model the SC and this adversarial interaction using probabilistic Moore machines~\cite{rabin1963probabilistic} and quantify security through Differential Privacy (DP)~\cite{DBLP:conf/icalp/Dwork06}. 
DP is a natural fit for this setting because it directly bounds the information leakage from observable outputs and provides a compositional worst-case guarantee~\cite{WassermanZhou2010,KairouzOhViswanath2015,DR:14:AFDP}.
By deriving an optimal attack strategy, we show that current PSC designs fail to provide sufficient DP guarantees under worst-case conditions, rendering the attack always successful.

To overcome these limitations, we propose a novel PSC design framework with additional probabilistic transitions.
Given a target privacy budget $\epsilon > 0$ (smaller $\epsilon$ means stronger privacy), our \emph{parameter synthesis problem} is: find the defense parameter $p \in (0,1)$ and update probability $m \in (0,1]$ such that the enhanced PSC satisfies $(\epsilon, 0)$-DP and the stationary misprediction rate is minimized.
We solve this problem analytically (Prop.~\ref{prop:optimalP}), yielding the optimal $p^{*} = \frac{1}{1+e^{\epsilon}}$.
Moreover, we assess performance impact by deriving an analytical solution for the misprediction rate based on the PSC's stationary distribution.
While we focus on the ubiquitous 2-bit counters, our methodology extends to arbitrary \kbit designs. 
Our main contributions are: 
(i) a formal framework for modeling saturating counters and the Prime+Probe attack using probabilistic Moore machines, which reveals inherent privacy vulnerabilities in existing PSC designs through the derivation of the optimal attack strategy;
(ii) a robust PSC design suitable for parameter synthesis, enabling the enforcement of rigorous differential privacy guarantees during the design phase;
(iii) the stationary misprediction rate as a performance metric and its analytical derivation; and
(iv) a formal privacy-utility trade-off analysis, including optimal parameter selection and a comparison with the randomized response mechanism that demonstrates the advantage of our targeted randomization.
Empirical results confirm that our enhanced PSCs achieve a superior balance of provable security and high performance compared to traditional and existing probabilistic counterparts.
Compared to Zhao et al.~\cite{JCST2021zhao}, who introduced PSCs with empirical security evaluation, we provide the \emph{first formal DP analysis} with provable guarantees, identify a fundamental vulnerability in the original PSC (the $c{=}1$ observation), propose an enhanced PSC achieving pure $\epsilon$-DP for any target budget, and derive the stationary misprediction rate in closed form.

\paragraph{Related work.}
Side-channel attacks against branch predictors and shared micro-architectural resources (e.g., SGX, BPUs, and caches) have been extensively studied~\cite{DBLP:conf/uss/BulckWKPS17,DBLP:conf/sp/XuCP15,DBLP:conf/uss/0001SGKKP17,evtyushkin2018branchscope,DBLP:journals/ieeesp/ChenCXZLL20,DBLP:conf/woot/BrasserMDKCS17,DBLP:conf/eurosec/GotzfriedESM17,DBLP:conf/usenix/HahnelCP17,DBLP:conf/dimva/SchwarzWGMM17,DBLP:conf/uss/DisselkoenKPT17,DBLP:conf/dimva/GrussMWM16,DBLP:conf/dac/KayaalpAPJ16,DBLP:journals/cacm/KocherHFGGHHLMP20,DBLP:conf/ccs/MaisuradzeR18}. 
Most relevant to our work, Wang et al.~\cite{DBLP:conf/hpca/WangTXW24} utilized symbolic execution to identify attack patterns on predictors. 
However, their evaluation of probabilistic saturating counters remained largely empirical. 
In contrast, we provide a theoretical framework that formally explains the effectiveness of attacks and explicitly models how attack outputs lead to branch inference, offering a more rigorous security assessment.

Differential privacy~\cite{DBLP:conf/icalp/Dwork06} has been widely adopted across various fields, including economics, healthcare, and deep learning~\cite{DBLP:conf/focs/McSherryT07,DBLP:conf/innovations/NissimST12,DBLP:journals/csur/ZhangZXZY22,DBLP:conf/edbt/DankarE12,DBLP:journals/corr/abs-1910-02578,DBLP:journals/patterns/DydaPCFPRWMHWL21,DBLP:conf/ccs/AbadiCGMMT016,DBLP:journals/corr/abs-1911-11607,DBLP:conf/cloud2/ChengYHZ0LLC19,DBLP:journals/corr/abs-2007-11524,DBLP:conf/nips/GhaziGKMZ21,DBLP:journals/corr/abs-2202-05089,DBLP:conf/aaai/ChengWZCW022}.
Importantly, DP guarantees imply rigorous bounds on an adversary's ability to make inferences: Wasserman and Zhou~\cite{WassermanZhou2010} established a statistical framework connecting DP with hypothesis testing, and Kairouz et al.~\cite{KairouzOhViswanath2015} studied the composition of DP mechanisms and the resulting attack bounds.
From a formal verification perspective, DP can be specified and verified using probabilistic couplings in programming languages~\cite{DBLP:conf/popl/ZhangK17,DBLP:conf/pldi/WangDWKZ19,DBLP:conf/ccs/WangDKZ20}, automated theorem proving~\cite{DBLP:conf/ccs/BartheFGGHS16,DBLP:journals/toplas/BartheKOB13,DBLP:conf/popl/BartheGAHRS15}, and probabilistic model checking~\cite{DBLP:journals/tcs/LiuWFZ23,DBLP:conf/vmcai/LiuWZ22}. 
Our work extends these formal techniques to the hardware security domain, leveraging the connection between DP and inference bounds to provide meaningful security guarantees for PSCs.

\section{Preliminaries}
\label{sec:preliminary}
In this section, we introduce key concepts that will be used later and review side-channel attacks on deterministic saturating counters.
Throughout the paper we fix $\MooreInput = \MooreOutput = \setnocond{\taken, \notaken}$ as alphabets and we just omit them from the definitions.

\subsection{Saturating counters and Moore machines}

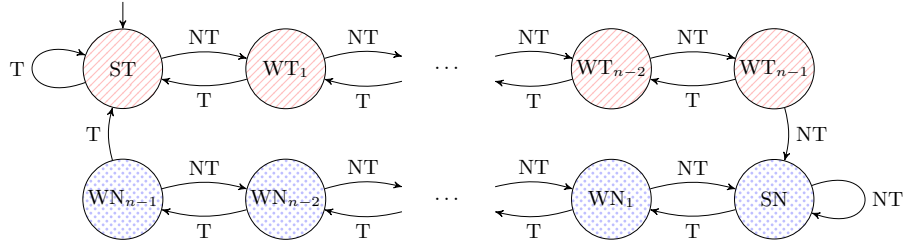
\begin{figure}[t]
    \centering
    \resizebox{\linewidth}{!}{
    \begin{tikzpicture}
        \node[Notakenode,fixedsize] (SN) at (0,0) {$\strongnotaken$};
        \node[Notakenode,fixedsize] (WN1) at ($(SN) - (2.5,0)$) {$\weaknotaken_{1}$};
        \node[fixedsize] (DN) at ($(WN1) - (2.5,0)$) {$\cdots$};
        \node[Notakenode,fixedsize] (WNn2) at ($(DN) - (2.5,0)$) {$\weaknotaken_{n-2}$};
        \node[Notakenode,fixedsize] (WNn1) at ($(WNn2) - (2.5,0)$) {$\weaknotaken_{n-1}$};

        \node[Takenode,fixedsize] (WTn1) at ($(SN) - (0,-2)$) {$\weaktaken_{n-1}$};
        \node[Takenode,fixedsize] (WTn2) at ($(WTn1) - (2.5,0)$) {$\weaktaken_{n-2}$};
        \node[fixedsize] (DT) at ($(WTn2) - (2.5,0)$) {$\cdots$};
        \node[Takenode,fixedsize] (WT1) at ($(DT) - (2.5,0)$) {$\weaktaken_{1}$};
        \node[Takenode,fixedsize, initial] (ST) at ($(WT1) - (2.5,0)$) {$\strongtaken$};

        \draw (SN) to[in=-20,out=20,looseness=7] node[right] {$\notaken$} (SN);
        \draw (SN) to[bend left=15] node[below] {$\taken$} (WN1);
        \draw (WN1) to[bend left=15] node[below] {$\taken$} (DN);
        \draw (WN1) to[bend left=15] node[above] {$\notaken$} (SN);
        \draw (DN) to[bend left=15] node[below] {$\taken$} (WNn2);
        \draw (DN) to[bend left=15] node[above] {$\notaken$} (WN1);
        \draw (WNn2) to[bend left=15] node[below] {$\taken$} (WNn1);
        \draw (WNn2) to[bend left=15] node[above] {$\notaken$} (DN);
        \draw (WNn1) to[bend left=15] node[left] {$\taken$} (ST);
        \draw (WNn1) to[bend left=15] node[above] {$\notaken$} (WNn2);
        
        \draw (ST) to[in=160,out=200,looseness=7] node[left] {$\taken$} (ST);
        \draw (ST) to[bend left=15] node[above] {$\notaken$} (WT1);
        \draw (WT1) to[bend left=15] node[above] {$\notaken$} (DT);
        \draw (WT1) to[bend left=15] node[below] {$\taken$} (ST);
        \draw (DT) to[bend left=15] node[above] {$\notaken$} (WTn2);
        \draw (DT) to[bend left=15] node[below] {$\taken$} (WT1);
        \draw (WTn2) to[bend left=15] node[above] {$\notaken$} (WTn1);
        \draw (WTn2) to[bend left=15] node[below] {$\taken$} (DT);
        \draw (WTn1) to[bend left=15] node[right] {$\notaken$} (SN);
        \draw (WTn1) to[bend left=15] node[below] {$\taken$} (WTn2);
    \end{tikzpicture}
    }
    \caption{A \kbit saturating counter, with a total of $2n=2^{k}$ states. Red-striped states output $\taken$, while blue-dotted states output $\notaken$.}
    \label{fig:kbitSaturatingCounter}
\end{figure}

A Saturating Counter (SC) records the branch history and predicts the next branch resolution.
The general structure of a \kbit SC is shown in Fig.~\ref{fig:kbitSaturatingCounter} and has $2^k$ states: 
one Strongly and $2^{k-1} - 1$ Weakly Not-taken blue-dotted states and, symmetrically, $2^{k-1} - 1$ Weakly and one Strongly Taken red-striped states. 
The predictor operates with two input symbols: $\taken$ and $\notaken$, indicating whether a branch is Taken or Not-Taken during a program execution, respectively. 
A \emph{misprediction} occurs when the prediction (e.g., $\taken$ in states $\strongtaken$ or $\weaktaken$) conflicts with the actual execution outcome (e.g., $\notaken$). 
The SC requires $n=2^{k-1}$ consecutive mispredictions to alter for sure its prediction direction due to the hysteresis provided by the strong/weak (S/W) bit, which allows the SC to adjust its state based on the actual execution direction.
An SC can be formally modeled as a Moore machine.
\begin{definition}
\label{def:mooreMachine}
    A \emph{Moore machine} is as a tuple $M = (\mstates, \minit, \MooreTransition, \Outfunc)$, where 
    $\mstates$ is a finite set of states, 
    $\minit \in \mstates$ is the initial state, 
    $\MooreTransition \colon \mstates \times \MooreInput \to \mstates$ is the transition function that maps the current state and the input to the next state, 
    and 
    $\Outfunc \colon \mstates \to \MooreOutput$ is the output function that maps a state to an output symbol.
\end{definition}
The Moore machine corresponding to the \kbit saturating counter shown in Fig.~\ref{fig:kbitSaturatingCounter} has the depicted nodes as states, the initial state is $\strongtaken$, as indicated by the sourceless incoming arrow, the transition function is given by the edges between nodes, and the output function assigns $\taken$ to the red-striped states and $\notaken$ to the blue-dotted states. 
A 2-bit saturating counter has only four states: $\strongtaken$, $\weaktaken$, $\weaknotaken$, and $\strongnotaken$, i.e., the states on the left and right of the counter shown in Fig.~\ref{fig:kbitSaturatingCounter}.

In addition to ordinary Moore machines, we will also use \textit{probabilistic Moore machines} to model branch predictors when probabilities are introduced.
\begin{definition}\label{def:mc}
  A \emph{probabilistic Moore machine} is a tuple $\mc = (\mstates, \initdistribution, \mtransitions, \Outfunc)$, where
  $\mstates$ and $\Outfunc$ are as in Def.~\ref{def:mooreMachine}, 
  $\initdistribution \colon \mstates \to [0,1]$ is the initial distribution such that $\sum_{s \in \mstates} \initdistribution(s) = 1$, 
  and 
  $\mtransitions \colon \mstates \times \MooreInput \times \mstates \to [0, 1]$ is the probability transition function such that $\sum_{t \in \mstates} \mtransitions(s, \sigma, t) = 1$ for each state $s \in \mstates$ and $\sigma \in \MooreInput$.
\end{definition}
Let $P_{s, \sigma}$ be the probability of receiving the input symbol $\sigma \in \MooreInput$ when being in state $s \in \mstates$.
These probabilities are specified externally by the workload or by the phase of the attack, and are not determined by Def.~\ref{def:mc}.
Then we can represent the transition probability function $\mtransitions$ as the matrix $\tranmatrix = (\mtransitions(s, t))_{\mstates \times \mstates}$ where each entry $\tranmatrix(s, t)$ is defined as $\tranmatrix(s, t) = \sum_{\sigma \in \MooreInput} \mtransitions(s, \sigma, t) \cdot P_{s, \sigma}$.
A probability distribution $\mu$ on $\mstates$ is stationary if $\mu  \tranmatrix=\mu$. 


\subsection{Prime+Probe attacks on branch predictors}

\begin{algorithm}[t]
    \begin{algorithmic}[1]
        \Statex \textbf{Input:} Number of branch executions $n$; the victim thread's direction $v \in \{\taken, \notaken\}$
        \Statex \textbf{Output:} The number of probes before reaching the cut-off point
        \Function{FindCutOffPoint}{$n,v$}
            \State execute the attacker branch with direction $\taken$ for $n$ times 
            \label{alg:cutoff:prime}
            \State execute the victim branch with direction $v$
            \label{alg:cutoff:attach}
            \State $c \gets 0$
            \label{alg:cutoff:probeStart}
            \While{$c \leq n$}
                \State execute the attacker branch with direction $\notaken$
                \If{prediction is correct}
                    \Return $c$
                \EndIf
                \State{$c \gets c+1$}
            \EndWhile
            \label{alg:cutoff:probeEnd}
            \State \Return $c$
        \EndFunction
    \end{algorithmic}
    \caption{Algorithm for finding the cut off point}
    \label{alg:cutoff}
\end{algorithm}

Branch predictor side-channel attacks are based on the ability to prime and probe them. 
The predictor encodes the execution history of branch instructions, potentially including sensitive branches; 
an attacker cannot observe the internal state of the predictor, but can observe its output. 
By providing appropriate inputs, the attacker can guide the predictor to a known state and then infer the internal state after the victim has taken its branch direction.
In \cite{JCST2021zhao,evtyushkin2018branchscope}, this is achieved by identifying the \textit{cut-off point} of the saturating counter: 
the first instance of a correct prediction following multiple mispredictions. 
By locating this cut-off point, the attacker can deduce whether the victim thread's execution direction was $\taken$ or $\notaken$.
The cut-off point can be found by means of Alg.~\ref{alg:cutoff} that we can divide into three phases:
\begin{itemize}
\item 
    Phase~1 (line~\ref{alg:cutoff:prime}): 
    the attacker first primes the target saturating counter to a predefined state, by repeatedly issuing $\taken$ to take the branch at least $n$ times; 
    if $n \geq 2^{k-1}$, then state $\strongtaken$ is for sure reached.
\item 
    Phase~2 (line~\ref{alg:cutoff:attach}): 
    the victim thread executes its branch with $v = \taken$ or $v = \notaken$; 
    the counter either remains in $\strongtaken$ or moves to $\weaktaken_{1}$, respectively.
\item 
    Phase~3 (lines~\ref{alg:cutoff:probeStart}--\ref{alg:cutoff:probeEnd}):
    the attacker probes the predictor by counting how many $\notaken$ branch executions are needed to reach $\strongnotaken$, the first state where the prediction agrees with the input $\notaken$.
\end{itemize}
The core of distinguishing the victim's behavior lies in observing differences in prediction outcomes during the probing phase. 
A critical observation is the \textit{cut-off point} in Phase~3, which corresponds to the $\strongnotaken$ state of the predictor. 
Before this point, only mispredictions are observed; afterwards, all predictions are correct. 
By counting the number of mispredictions $c$ before reaching the cut-off point, the attacker can infer the direction of the victim's branch, by comparing $c$ with the known value $2^{k-1}$: 
if $c = 2^{k-1}$, then the victim's direction was $v = \taken$; 
if $c < 2^{k-1}$, then the victim's direction was $v = \notaken$.
Through repeated execution of this algorithm, the attacker can continuously decode the victim's sensitive branch executions, posing significant security risks.

\subsection{Differential privacy} 

To measure how much a branch predictor is resistant to attacks we use differential privacy~\cite{D:06:DP,NRS:07:dp,DR:14:AFDP}, a framework designed for the protection of individual data during publication and analysis.
Given a set of \emph{data entries} $\mathcal{X}$, a \emph{dataset} of size $n$ is an element of $\mathcal{X}^{n}$.
Two datasets $\mathcal{D}, \mathcal{D}' \in \mathcal{X}^{n}$ are \emph{neighbors}, denoted by $\Delta(\mathcal{D}, \mathcal{D}') \leq 1$, if they differ in at most one data entry. 
To ensure individual privacy, random noise is introduced to the dataset so that true values cannot be disclosed or inferred. 
A \emph{data publishing mechanism} $\mathcal{M}$ is a randomized algorithm that takes a dataset $\mathcal{D}$ as input; 
let $\textmd{range}(\mathcal{M})$ denote the set of all possible outputs of $\mathcal{M}$ and assume we are given a probability space over $\textmd{range}(\mathcal{M})$. 
Differential privacy is satisfied by $\mathcal{M}$ if its output distributions over all neighboring datasets are close enough. 
\begin{definition}
\label{def:differentialPrivacy}
    Given $\epsilon, \delta \in \posreals$, a mechanism $\mathcal{M}$ is called \emph{$(\epsilon, \delta)$-differentially private} if for all measurable sets $O \subseteq \textmd{range}(\mathcal{M})$ and all datasets $\mathcal{D}, \mathcal{D}' \in \mathcal{X}^{n}$ such that $\Delta(\mathcal{D}, \mathcal{D}') \leq 1$, it holds that 
    $\Pr(\mathcal{M}(\mathcal{D}) \in O) \leq e^{\epsilon} \cdot \Pr(\mathcal{M}(\mathcal{D}') \in O) + \delta$.
\end{definition}
The parameters $\epsilon$ and $\delta$ control the probability differences between the outputs in any two neighboring input datasets regarding an output set $O$. Evidently, smaller values of these parameters result in smaller discrepancies in the probability distributions, thus offering enhanced privacy protections. 
The special case $\delta = 0$ is known as \emph{pure} $\epsilon$-differential privacy, a strictly stronger guarantee than $(\epsilon,\delta)$-DP with $\delta > 0$~\cite{DR:14:AFDP}.
In subsequent sections, we adapt this definition to safeguard sensitive information of the victim thread within the context of branch predictors.
Concretely, the ``dataset'' is the victim's single branch direction ($v \in \{\taken, \notaken\}$); two ``neighboring datasets'' differ only in this one bit; the ``mechanism'' is the Prime+Probe attack output (the cut-off point $c$); and DP guarantees that the distribution of $c$ is nearly the same regardless of $v$.
This is a \emph{local} security guarantee for the PSC primitive under the Prime+Probe observation model.
We note that alternative quantitative security notions, such as quantitative information flow~\cite{DBLP:conf/fossacs/Smith09}, could also be applied; DP is particularly well-suited here because (1)~it provides a compositional worst-case guarantee for arbitrary post-processing, (2)~it directly bounds the adversary's inference advantage as $\frac{e^{\epsilon}}{1+e^{\epsilon}}$, and (3)~the resulting constraints on parameters $m$ and $p$ admit efficient synthesis.

\section{Security Analysis of Existing PSCs}
\label{sec:existingPSCsecurityAnalysis}

In this section, we first review the existing PSCs presented in~\cite{JCST2021zhao}. 
We then provide two contributions related to these PSCs.
First, we employ probabilistic Moore machines to formally model both the PSC mechanism and its associated attack surface, proposing a rigorous security specification grounded in differential privacy.
Second, by deriving the optimal attack strategy and quantifying the probability of a successful attack, we conduct a comprehensive theoretical analysis that reveals a critical security vulnerability in the existing PSC design.

\subsection{Differential privacy analysis of PSCs}

\begin{figure}[t]
    \centering
    \begin{tikzpicture}
        \path[use as bounding box] (-2.2,1) rectangle (8.2,-3);
        \node[Takenode, initial] (ST) at (0,0) {$\strongtaken$};
        \node[Takenode] (WT) at ($(ST) + (6,0)$) {$\weaktaken$};
        \node[Notakenode] (SN) at ($(WT) - (0,2)$) {$\strongnotaken$};
        \node[Notakenode] (WN) at ($(SN) - (6,0)$) {$\weaknotaken$};
        
        \draw (ST) to[bend left=12.5] node {$\notaken/m$} (WT);
        \draw (ST) to[in=160,out=200,looseness=7] node[left] {$\taken/1$} (ST);
        \draw (ST) to[in=110,out=70,looseness=7] node[pos=0.85,left] {$\notaken/1{-}m$} (ST);
        
        \draw (WT) to[bend left=12.5] node {$\notaken/m$} (SN);
        \draw (WT) to[bend left=12.5] node[below] {$\taken/m$} (ST);
        \draw (WT) to[in=20,out=-20,looseness=7] node[right] {$\taken/1{-}m$} (WT);
        \draw (WT) to[in=70,out=110,looseness=7] node[pos=0.85] {$\notaken/1{-}m$} (WT);
        
        \draw (SN) to[bend left=12.5] node[below] {$\taken/m$} (WN);
        \draw (SN) to[in=20,out=-20,looseness=7] node[right] {$\taken/1{-}m$} (SN);
        \draw (SN) to[in=290,out=250,looseness=7] node[pos=0.85,right] {$\notaken/1$} (SN);
        
        \draw (WN) to[bend left=12.5] node[above] {$\notaken/m$} (SN);
        \draw (WN) to[bend left=12.5] node[left] {$\taken/m$} (ST);
        \draw (WN) to[in=160,out=200,looseness=7] node[left] {$\taken/1{-}m$} (WN);
        \draw (WN) to[in=250,out=290,looseness=7] node[pos=0.85] {$\notaken/1{-}m$} (WN);
    \end{tikzpicture}
    \caption{A 2-bit PSC, shown with initial state $\strongtaken$ after the attacker's priming phase}
    \label{fig:2bitPSC}
\end{figure}
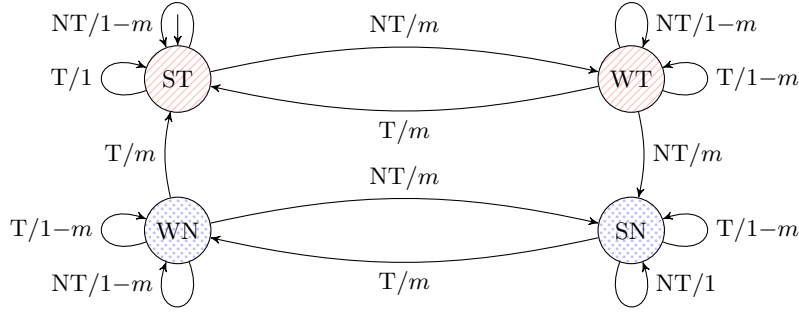

The attack described in Alg.~\ref{alg:cutoff} exploits the deterministic nature of classical saturating counters, where the next state is uniquely determined by the current state and input. 
To defend against such attacks, Zhao et al.~\cite{JCST2021zhao} proposed to use probabilistic saturating counters. 
This approach introduces a probabilistic threshold into the counter mechanism: 
upon a state transition, a random number is generated and compared to the threshold; 
if the number is below the threshold, the transition proceeds. 
Otherwise, the counter remains in its current state.
Taking the 2-bit PSC illustrated in Fig.~\ref{fig:2bitPSC} as an example, given the threshold $m \in [0,1]$, for each state $s$ and input letter $b \in \setnocond{\taken, \notaken}$, with probability $m$ the PSC transitions from $s$ to the intended target state of $b$;
with the remaining probability $1 - m$ the PSC ignores the input $b$ and remains in $s$.

By adding the threshold $m$, the attacker cannot precisely predict the state of the counter due to the probabilistic nature of the steps leading to the cut-off point. 
The experimental results given in~\cite{JCST2021zhao} show the success rate of the attack under various values of $m$, yet a theoretical analysis is essential to understand the optimal attack strategies and the misprediction rates, crucial to performance evaluation. 
We apply differential privacy to establish theoretical guarantees for defense mechanisms in PSCs. Additionally, we analyze the optimal attack strategy for Alg.~\ref{alg:cutoff}, identifying scenarios where the attacker can successfully infer the victim's execution. 
To mitigate this, we propose a new PSC model and calculate its misprediction rate;
we begin by defining differential privacy for PSCs.

In the context of branch prediction side-channel attacks, the sensitive data corresponds to the direction of a single victim's branch. 
Therefore, we refine the standard notion of neighboring datasets to reflect this specific threat model: 
two datasets (i.e., execution histories) are considered neighbors if they differ only for the outcome of the single victim's branch.
As data publishing mechanism, we consider the attacker based on \functionName{FindCutOffPoint} but we reverse the meaning of its outcome:
the mechanism provides privacy if the attacker fails in computing the right cut-off point with large enough probability.
This results in the following definition of $(\epsilon, \delta)$-differential privacy for PSCs.
\begin{definition}
\label{def:pscdp}
    Given $n \in \naturals$, let $\mathit{out}$ be the random output returned by Alg.~\ref{alg:cutoff} on input $n$ and victim direction $v$.
    A PSC satisfies $(\epsilon, \delta)$-differential privacy if, for every output $c$ of Alg.~\ref{alg:cutoff}, the probability of observing $c$ satisfies
    \begin{align*}
    & \Pr[\mathit{out} = c | v = \taken] \leq \mathrm{e}^\epsilon \Pr[\mathit{out} = c | v = \notaken] + \delta, \\
    & \Pr[\mathit{out} = c | v = \notaken] \leq \mathrm{e}^\epsilon \Pr[\mathit{out} = c | v = \taken] + \delta.
    \end{align*}
\end{definition}
In concrete scenario attacks, we assume the target branch has been executed sufficiently many times with direction $\taken$ that the probability of the counter being in a state other than $\strongtaken$ is in practice negligible (cf.\@ the experiments in Sect.~\ref{sec:experiments}).

This definition provides a natural security metric: the smaller the achievable $\epsilon$, the harder it is for the attacker to distinguish between the two branch directions.
When $\epsilon = 0$ and $\delta = 0$, the output distributions are identical and the attacker gains no information---corresponding to perfect privacy at the cost of a $50\%$ misprediction rate.
When $\epsilon$ is large, the output distributions diverge significantly, allowing the attacker to infer the branch direction with high probability.
Our goal in the remainder of this section is to determine the achievable $\epsilon$ for the existing PSC design and, in Sect.~\ref{sec:enhancedPSC}, to synthesize an enhanced PSC that satisfies a given $\epsilon$-DP requirement.

\subsection{Modeling attacks to PSCs with probabilistic Moore machines}
\label{ssec:PSCattackModel}

The concept of differential privacy on PSCs guides us in analyzing the probability of obtaining each output under different executions of the victim's branch. 
Specifically, given an observation $c$, our goal is to determine the attacker's strategy to infer whether the branch was taken or not.
In~\cite{JCST2021zhao}, experiments fix a parameter $m$ and simulate the victim's behavior multiple times. 
The attacker records the frequency of different outputs when the branch is taken versus not-taken, then guesses the victim's action based on the most frequent outcome. 
However, the limitation of this method is that it requires a large number of tests and does not provide general conclusions for arbitrary values of~$m$.

To address this limitation, we model the attack algorithm using parametric Moore machines, as they allow us to derive the optimal attack strategy. 
This approach allows for more generalized insight regardless of the specific value of $m$.
During the three phases of Alg.~\ref{alg:cutoff}, the directions of the branches executed are different, as well as their probabilities in the different states $s$:
\begin{itemize}
\item 
    In Phase~1, the branch is always taken. 
    Thus, $P_{s, \taken} = 1$ and $P_{s, \notaken} = 0$.
\item 
    In Phase~2, the branch is executed once depending on the victim's thread actual direction $v$, so $P_{s, v} = 1$ and $P_{s, \neg v} = 0$ for the opposite direction $\neg v$.
\item 
    In Phase~3, the branch is always not-taken. 
    Thus, $P_{s, \taken} = 0$ and $P_{s, \notaken} = 1$.
\end{itemize}

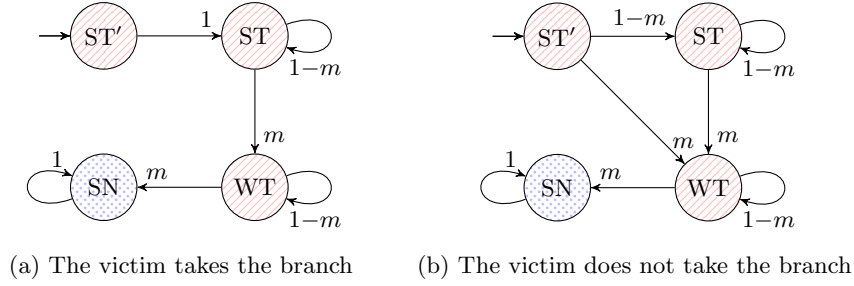
\begin{figure}[t]
    \centering
    \begin{tikzpicture}
        \path[use as bounding box] (-2.3,1.55) rectangle (8.9,-2.1);
        \node[label={below:(a) The victim takes the branch}] (BT) at (0,0) {
            \begin{tikzpicture}
                \node[Takenode,initial, initial left] (STp) at (0,0) {$\strongtaken'$};
                \node[Takenode] (ST) at ($(STp) + (2,0)$) {$\strongtaken$};
                \node[Notakenode] (SN) at ($(STp) + (0,-2)$) {$\strongnotaken$};
                \node[Takenode] (WT) at ($(SN) + (2,0)$) {$\weaktaken$};
    
                \draw (STp) to node[pos=1,anchor=south east] {$1$} (ST);
                \draw (ST) to[in=-20,out=20,looseness=7] node[pos=0.8,below] {$1{-}m$} (ST);
                \draw (ST) to node[pos=0.8,right] {$m$} (WT);
                \draw (WT) to[in=-20,out=20,looseness=7] node[pos=0.8,below] {$1{-}m$} (WT);
                \draw (WT) to node[pos=1,anchor=south west] {$m$} (SN);
                \draw (SN) to[in=160,out=200,looseness=7] node[pos=1,anchor=south east] {$1$} (SN);            
            \end{tikzpicture}
        };
        \node[label={below:(b) The victim does not take the branch}] (BN) at ($(BT) + (6,0)$) {
            \begin{tikzpicture}
                \node[Takenode,initial, initial left] (STp) at (0,0) {$\strongtaken'$};
                \node[Takenode] (ST) at ($(STp) + (2,0)$) {$\strongtaken$};
                \node[Notakenode] (SN) at ($(STp) + (0,-2)$) {$\strongnotaken$};
                \node[Takenode] (WT) at ($(SN) + (2,0)$) {$\weaktaken$};
    
                \draw (STp) to node[pos=1,anchor=south east] {$1{-}m$} (ST);
                \draw (STp) to node[pos=0.8,right] {$m$} (WT);
                \draw (ST) to[in=-20,out=20,looseness=7] node[pos=0.8,below] {$1{-}m$} (ST);
                \draw (ST) to node[pos=0.8,right] {$m$} (WT);
                \draw (WT) to[in=-20,out=20,looseness=7] node[pos=0.8,below] {$1{-}m$} (WT);
                \draw (WT) to node[pos=1,anchor=south west] {$m$} (SN);
                \draw (SN) to[in=160,out=200,looseness=7] node[pos=1,anchor=south east] {$1$} (SN);            
            \end{tikzpicture}
        };
    \end{tikzpicture}
    \caption{The probabilistic Moore machines for the attack algorithm}
    \label{fig:2bitPSCprobMooreMachines}
\end{figure}
When $m=0$, only self-loops exist, so $c$ is independent of~$v$ and perfect privacy is achieved---at the cost of a $50\%$ misprediction rate.
When $m=1$, the PSC reduces to the deterministic SC, and the attacker can distinguish $v$ from the number of Phase~3 steps ($3$ for taken vs.~$2$ for not-taken).
For $0 < m < 1$, the probabilistic transitions make the probe count $c$ uncertain; the attacker can only infer $v$ with some probability.
For instance, with $m=0.5$, observing $c=1$ is impossible when the victim takes the branch but possible when it does not take the branch, already indicating a structural asymmetry in the two output distributions.
Our goal is to compare $\Pr[\mathit{out}=c \mid v=\taken]$ and $\Pr[\mathit{out}=c \mid v=\notaken]$ for each~$c$ to derive the optimal strategy.
To achieve our goal, we model Alg.~\ref{alg:cutoff} for both taken and not-taken scenarios of the victim's branch execution;
the corresponding probabilistic models are shown in Fig.~\ref{fig:2bitPSCprobMooreMachines} and they encode both the attacker and the victim behaviors at the beginning of the victim's Phase~2 execution: 
the transition from $\strongtaken'$ in each machine represents the victim's branch execution, reaching the appropriate state of the PSC $\strongtaken$ or $\weaktaken$ from where we simulate Phase~3 of Alg.~\ref{alg:cutoff}. 
Formally, let $M_{\taken}$ and $M_{\notaken}$ represent the transition matrices of the probabilistic Moore machine when the branch is taken or not-taken, respectively. 
Let $\indvec{\strongtaken}, \indvec{\strongnotaken} \in \posreals^{4}$ be two column vectors, with value $0$ everywhere except for value $1$ in the entry corresponding to state $\strongtaken$ in $\indvec{\strongtaken}$ and to state $\strongnotaken$ in $\indvec{\strongnotaken}$, respectively. 
Then $\transpose{\indvec{\strongtaken}} M_{\taken}^{c+1} \indvec{\strongnotaken}$ gives the probability of being in the absorbing state $\strongnotaken$ after $c+1$ steps when the victim takes the branch, which coincides with observing the cut-off point $c$, 
that is,
$\Pr[\mathit{out} = c | v = \taken] = \transpose{\indvec{\strongtaken}} M_{\taken}^{c+1} \indvec{\strongnotaken}$.
Similarly,
$\Pr[\mathit{out} = c | v = \notaken] = \transpose{\indvec{\strongtaken}} M_{\notaken}^{c+1} \indvec{\strongnotaken}$.
These probabilities are the conditional probabilities required to ensure differential privacy in PSCs as defined in Def.~\ref{def:pscdp}.

\subsection{Optimal attack strategy}
\label{ssect:Optimal}

The attacker's strategy is based on comparing the probabilities $\Pr[\mathit{out} = c | v = \taken]$ and $\Pr[\mathit{out} = c | v = \notaken]$ and selecting the branch direction with the highest probability. 
Specifically, we evaluate whether $c > 1/m$ satisfies the inequality
\begin{equation}
\label{eq:attackerStrategyCondition}
    \transpose{\indvec{\strongtaken}} M_{\taken}^{c+1} \indvec{\strongnotaken} > \transpose{\indvec{\strongtaken}} M_{\notaken}^{c+1} \indvec{\strongnotaken}.
\end{equation}
This defines the attacker's optimal strategy: 
if $c = \functionName{FindCutOffPoint}(n,v)$ satisfies $c > 1/m$, then it infers that $v = \taken$; 
otherwise, it concludes that  $v = \notaken$.
In fact, the expected number of steps from $\weaktaken$ to $\strongnotaken$ is $1/m$, so having $c > 1/m$ hints that the PSC was in state $\strongtaken$ after Phase~2 of Alg.~\ref{alg:cutoff} instead of being in $\weaktaken$, thus it is more likely that the victim took the branch (i.e., $v = \taken$).
%
Regardless of the value of $m$, if $c = 1$ is observed, the attack succeeds with probability $1$, indicating that no $(\epsilon, \delta)$-differential privacy can be achieved when $\delta < 1$. 
This limitation arises from the PSC structure: 
if the victim takes the branch, at least two steps are needed to reach $\strongnotaken$ from $\strongtaken$, unlike the case where the branch is not taken, as one step suffices to reach $\strongnotaken$. 
We design an enhanced PSC with improved security to address this issue in Sect.~\ref{sec:enhancedPSC}.

We now quantify the success probability of this strategy.
Bayes' theorem gives the posterior probability of the attacker's guess being correct:

\[
    \Pr[v{=}b|\mathit{out}{=}c] =
    \frac{\Pr[\mathit{out}{=}c|v{=}b]\Pr[v{=}b]}
    {\sum_{b' \in \{\taken,\notaken\}} \Pr[\mathit{out}{=}c|v{=}b']\Pr[v{=}b']},
\]
where $b$ is the attacker's guess ($\taken$ if $c > 1/m$, $\notaken$ otherwise).
When the attacker has no prior knowledge, we use the uniform prior $\Pr[v{=}\taken] = \Pr[v{=}\notaken] = 0.5$, and the prior factor cancels from the numerator and denominator.
The likelihoods come from the matrix powers $M_{\taken}^{c+1}$ and $M_{\notaken}^{c+1}$ (Sect.~\ref{ssec:PSCattackModel}); when one likelihood dominates the other, the posterior approaches~$1$ and the attack succeeds with near certainty.

The connection to differential privacy is immediate.
If the PSC provides $(\epsilon,0)$-differential privacy, Def.~\ref{def:pscdp} gives $\Pr[\mathit{out}{=}c|v{=}\taken] \leq \mathrm{e}^{\epsilon} \cdot \Pr[\mathit{out}{=}c|v{=}\notaken]$, yielding the well-known upper bound on the success probability~\cite{WassermanZhou2010,KairouzOhViswanath2015}:
\[
    \Pr[v{=}b|\mathit{out}{=}c] \leq \frac{\mathrm{e}^{\epsilon}}{1 + \mathrm{e}^{\epsilon}} = \frac{1}{1 + \mathrm{e}^{-\epsilon}}.
\]
This bound applies specifically to pure $\epsilon$-differential privacy ($\delta = 0$); for general $(\epsilon,\delta)$-DP with $\delta > 0$, the bound requires additional correction terms~\cite{KairouzOhViswanath2015}.
As $\epsilon \to 0$ this bound approaches $0.5$ (random guessing, meaning the defense is effective); a large $\epsilon$ allows it to approach~$1$ (the attacker can determine the branch direction with near certainty).
This provides a clear quantitative relationship: the privacy budget $\epsilon$ directly governs the worst-case adversarial advantage.

In summary, the attacker's strategy operates as follows: for each observed probe count $c$, the attacker computes the posterior probability $\Pr[v{=}b \mid \mathit{out}{=}c]$ using Bayes' theorem and selects the direction $b$ with the higher posterior.
The key insight is that the PSC's transition structure creates an asymmetry in the distribution of $c$ under different victim directions: when the victim takes the branch, the PSC remains at $\strongtaken$ and requires more $\notaken$ steps to reach $\strongnotaken$, leading to larger $c$ values; when the victim does not take, the PSC moves to $\weaktaken$ (or stays at $\strongtaken$ with reduced probability in the enhanced model), leading to smaller $c$ values.
This structural asymmetry is precisely what the attacker exploits: the ratio between $\Pr[\mathit{out}{=}c \mid v{=}\taken]$ and $\Pr[\mathit{out}{=}c \mid v{=}\notaken]$ grows as $c$ deviates from the expected transition count $1/m$, enabling the attacker to infer the branch direction with increasing confidence.
The DP guarantee ensures that this ratio, which means the attacker's advantage, is bounded by $\epsilon$.

\section{Enhanced PSC Model and Parameter Synthesis}
\label{sec:enhancedPSC}

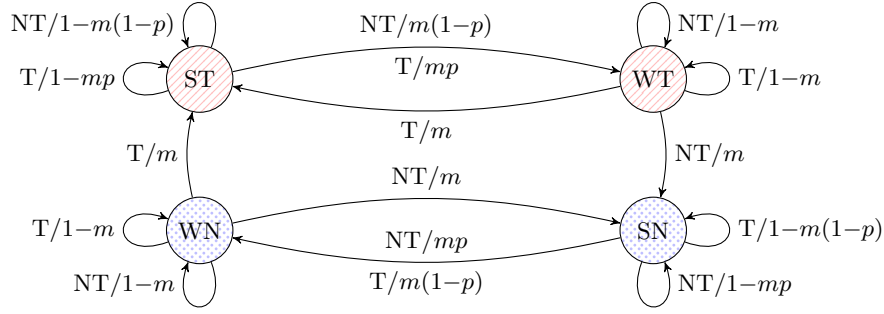
\begin{figure}[t]
\centering
    
    
    
        
    \begin{tikzpicture}
        \path[use as bounding box] (-2.2,1) rectangle (8.2,-3);
        \node[Takenode] (ST) at (0,0) {$\strongtaken$};
        \node[Takenode] (WT) at ($(ST) + (6,0)$) {$\weaktaken$};
        \node[Notakenode] (SN) at ($(WT) - (0,2)$) {$\strongnotaken$};
        \node[Notakenode] (WN) at ($(SN) - (6,0)$) {$\weaknotaken$};
        
        \draw (ST) to[bend left=12.5] node[above] {$\notaken/m(1{-}p)$} node[below] {$\taken/mp$} (WT);
        \draw (ST) to[in=160,out=200,looseness=7] node[left] {$\taken/1{-}mp$} (ST);
        \draw (ST) to[in=110,out=70,looseness=7] node[pos=0.85,left] {$\notaken/1{-}m(1{-}p)$} (ST);
        
        \draw (WT) to[bend left=12.5] node {$\notaken/m$} (SN);
        \draw (WT) to[bend left=12.5] node[below] {$\taken/m$} (ST);
        \draw (WT) to[in=20,out=-20,looseness=7] node[right] {$\taken/1{-}m$} (WT);
        \draw (WT) to[in=70,out=110,looseness=7] node[pos=0.85] {$\notaken/1{-}m$} (WT);
        
        \draw (SN) to[bend left=12.5] node[below] {$\taken/m(1{-}p)$} node[above] {$\notaken/mp$}(WN);
        \draw (SN) to[in=20,out=-20,looseness=7] node[right] {$\taken/1{-}m(1{-}p)$} (SN);
        \draw (SN) to[in=290,out=250,looseness=7] node[pos=0.85,right] {$\notaken/1{-}mp$} (SN);
        
        \draw (WN) to[bend left=12.5] node[above] {$\notaken/m$} (SN);
        \draw (WN) to[bend left=12.5] node[left] {$\taken/m$} (ST);
        \draw (WN) to[in=160,out=200,looseness=7] node[left] {$\taken/1{-}m$} (WN);
        \draw (WN) to[in=250,out=290,looseness=7] node[pos=0.85] {$\notaken/1{-}m$} (WN);
    \end{tikzpicture}
    \caption{The probabilistic Moore machine for the enhanced PSC}
    \label{fig:enhancedPSC}
\end{figure}

As seen in the previous section, it is still possible for an attacker to accurately infer the branch execution of the victim because there exist deterministic transitions in the PSC.
In particular, the observation $c = 1$ always reveals that $v = \notaken$, since reaching $\strongnotaken$ in one step is impossible when the victim takes the branch.
This constitutes a fundamental vulnerability that cannot be mitigated by adjusting the existing parameter $m$ alone.
In this section, we design a new model based on the original PSC and establish its differential privacy properties by parameter synthesis.

\subsection{The enhanced PSC model}

We enhance the PSC model by introducing additional probabilistic choices. 
The original probability update mechanism is retained, where the transitions are updated with probability $m$. 
When a transition is triggered, the actual state reached by the transition is probabilistically decided by means of a new parameter $p \in [0,1]$, termed the \emph{defense parameter}: 
for state $\strongtaken$ and input $\taken$, instead of remaining in $\strongtaken$ for sure, with probability $p$ the PSC moves to $\weaktaken$ and for input $\notaken$, instead of going to $\weaktaken$ for sure, with probability $p$ the PSC remains in $\strongtaken$;
for state $\strongnotaken$ the defense is applied symmetrically.
The resulting enhanced PSC is shown in Fig.~\ref{fig:enhancedPSC}.
Notably, forcing all transitions to be probabilistic from both $\strongtaken$ and $\strongnotaken$ is essential to achieve security: 
as we will see below, they prevent an attack based on Alg.~\ref{alg:cutoff} similar to the one presented in Sect.~\ref{ssect:Optimal}.
Note that our enhanced PSC generalizes both the classical saturating counters (by setting $m = 1$ and $p = 0$) and the probabilistic saturating counters (when $p = 0$).


\begin{figure}[t]
    \centering
    \begin{tikzpicture}
        \path[use as bounding box] (-2.3,1.55) rectangle (8.9,-2.1);
        \node[label={below:(a) The victim takes the branch}] (BT) at (0,0) {
            \begin{tikzpicture}
                \node[Takenode,initial, initial left] (STp) at (0,0) {$\strongtaken'$};
                \node[Takenode] (ST) at ($(STp) + (2,0)$) {$\strongtaken$};
                \node[Notakenode] (SN) at ($(STp) + (0,-2)$) {$\strongnotaken$};
                \node[Takenode] (WT) at ($(SN) + (2,0)$) {$\weaktaken$};
    
                \draw (STp) to node[pos=1,anchor=south east] {$1{-}mp$} (ST);
                \draw (STp) to[left] node {$mp$} (WT);
                \draw (ST) to[in=-20,out=20,looseness=7] node[pos=1,anchor=north west] {$1{-}m(1{-}p)$} (ST);
                \draw (ST) to node[pos=0.8,right] {$m(1{-}p)$} (WT);
                \draw (WT) to[in=-20,out=20,looseness=7] node[pos=0.8,below] {$1{-}m$} (WT);
                \draw (WT) to node[pos=1,anchor=south west] {$m$} (SN);
                \draw (SN) to[in=160,out=200,looseness=7] node[pos=1,anchor=south east] {$1$} (SN);            
            \end{tikzpicture}
        };
        \node[label={below:(b) The victim does not take the branch}] (BN) at ($(BT) + (6,0)$) {
            \begin{tikzpicture}
                \node[Takenode,initial, initial left] (STp) at (0,0) {$\strongtaken'$};
                \node[Takenode] (ST) at ($(STp) + (3,0)$) {$\strongtaken$};
                \node[Notakenode] (SN) at ($(STp) + (0,-2)$) {$\strongnotaken$};
                \node[Takenode] (WT) at ($(SN) + (3,0)$) {$\weaktaken$};
    
                \draw (STp) to node[pos=1,anchor=south east] {$1{-}m(1{-}p)$} (ST);
                \draw (STp) to node[left] {$m(1{-}p)$} (WT);
                \draw (ST) to[in=-20,out=20,looseness=7] node[pos=1,anchor=north west] {$1{-}m(1{-}p)$} (ST);
                \draw (ST) to node[pos=0.8,right] {$m(1{-}p)$} (WT);
                \draw (WT) to[in=-20,out=20,looseness=7] node[pos=0.8,below] {$1{-}m$} (WT);
                \draw (WT) to node[pos=1,anchor=south west] {$m$} (SN);
                \draw (SN) to[in=160,out=200,looseness=7] node[pos=1,anchor=south east] {$1$} (SN);            
            \end{tikzpicture}
        };
    \end{tikzpicture}
    \caption{The enhanced probabilistic Moore machines for the attack algorithm}
    \label{fig:2bitEnhancedPSCprobMooreMachines}
\end{figure}
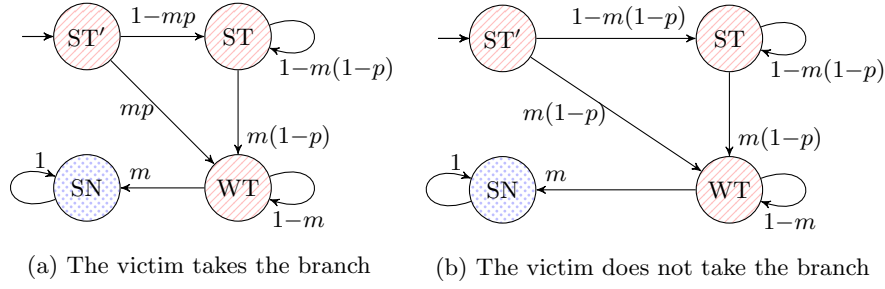

We now show how to synthesize parameters $m$ and $p$ such that the enhanced PSC satisfies a given $(\epsilon, \delta)$-differential privacy requirement. 
Similar to Sect.~\ref{ssec:PSCattackModel}, we consider the probabilistic Moore machine corresponding to the attack model, depicted in Fig.~\ref{fig:2bitEnhancedPSCprobMooreMachines}.
To ensure $(\epsilon, \delta)$-differential privacy, we must satisfy the requirements of Def.~\ref{def:pscdp}. 
This involves solving the two differential privacy bounds to determine the feasible range for $m$ and $p$. 
Given that the condition in Eq.~\eqref{eq:attackerStrategyCondition} underlying the attacker's strategy involves polynomials in $m$, $p$, and the exponent $c$, finding explicit analytic solutions is challenging. 
However, we  can fix one parameter and derive a suitable range for the other:
suppose that we are verifying the enhanced PSC for $(\epsilon, \delta)$-differential privacy with $\epsilon = 10^{-2}$ and $\delta = 10^{-5}$. 
We start by fixing the probability $m$ at $0.5$, which has shown good performance in~\cite{JCST2021zhao}. 
Using SageMath~\cite{zimmermann:18:SageMath} for symbolic algebra computation, we obtain $p \in [0.4899, 0.5017]$ fulfilling the differential privacy constraints.

We also establish a formal relationship between the defense parameter $p$ and the achievable privacy level, as presented in Thm.~\ref{thm:pTwoBitDifferentialPrivacyResult}:
\begin{restatable}{theorem}{pTwoBitDifferentialPrivacyResult}
\label{thm:pTwoBitDifferentialPrivacyResult}
    Given an enhanced PSC $\mc$ with defense parameter $p\in (0,1)$, 
    if $\privacy\geq \frac{1}{2}$,
    then $\mc$ satisfies $(\ln\frac{\privacy}{1-\privacy},0)$-differential privacy;
    otherwise, it satisfies $(\ln\frac{1-\privacy}{\privacy},0)$-differential privacy.
\end{restatable}

The detailed proof of Thm.~\ref{thm:pTwoBitDifferentialPrivacyResult} is provided in Appendix~\ref{appendix-formal}. 
A notable feature of this result is that the achievable differential privacy level depends solely on $p$ and is \emph{independent of $m$}. 
Intuitively, $m$ acts as a uniform ``speed'' parameter that scales all non-self-loop transition probabilities equally in both the taken and not-taken attack models (cf.\@ Fig.~\ref{fig:2bitEnhancedPSCprobMooreMachines}). 
The differential privacy guarantee is determined by the \emph{ratio} of the output probabilities $\Pr[\mathit{out}=c \mid v=\taken]$ and $\Pr[\mathit{out}=c \mid v=\notaken]$; 
since $m$ appears as a common multiplicative factor in both, it cancels out in this ratio via the binomial theorem (see the proof in Appendix~\ref{appendix-formal}). 
In contrast, $p$ introduces an \emph{asymmetry} exclusively at the strong states ($\strongtaken$/$\strongnotaken$), directly affecting how differently the two models behave and thus governing the privacy level.

Moreover, we remark that Thm.~\ref{thm:pTwoBitDifferentialPrivacyResult} establishes \emph{pure} $\epsilon$-differential privacy (i.e., $\delta = 0$) for all values of $p \in (0,1)$.
We adopted the more general $(\epsilon,\delta)$-DP framework (Def.~\ref{def:pscdp}) as the starting point of our analysis because it is the standard formulation in the DP literature~\cite{DR:14:AFDP} and, importantly, because we first needed to show that existing PSCs \emph{fail} to satisfy even this weaker notion (Sect.~\ref{ssect:Optimal}). 
The fact that our enhanced PSC achieves the stronger $\delta = 0$ guarantee is a positive outcome of the design.

The interplay between $p$ and the resulting privacy--utility trade-off is analyzed in detail in Sect.~\ref{ssec:privacy-utility}.

\subsection{Estimating the misprediction rate}
\label{ssec:counter-misprediction}

The branch execution of a program can be modeled as a sequence $\sigma$ in $\{\taken, \notaken\}^{*}$. 
Exact misprediction rates over all finite traces are generally infeasible to compute: traces grow exponentially, and real executions may be history-dependent. 
Worst-case sequences can force deterministic SCs to mispredict every step, but such traces are rare in practice; we therefore use a statistical estimate.

For a given branch, let the probability of executing $\taken$ be $t \in (0,1)$ and of $\notaken$ be $s = 1 - t$; let $q = 1 - p$.
Here, $t$ abstracts the empirical taken probability of one static branch, not a full generative model of program traces.
For $m > 0$ and non-degenerate branch behavior, the induced finite Markov chain is irreducible and aperiodic, and hence has a unique stationary distribution.
Let $\mu = [a,b,c,d]$ denote this stationary distribution over the states $\strongtaken$, $\weaktaken$, $\weaknotaken$, and $\strongnotaken$, respectively.
The stationary distribution is obtained by solving $\mu M = \mu$ for the transition matrix $M$ of Fig.~\ref{fig:enhancedPSC}; the misprediction rate is then $r = (a+b)s + (c+d)t$.
Solving these equations yields
\[
    r = \frac{ts(qt + ps)(1 + qs + pt) + ts(qs + pt)(1 + qt + ps)}{t(qt + ps)(1 + qs + pt) + s(qs + pt)(1 + qt + ps)}.
\]
Appendix~\ref{sec:rdetail} gives the full calculation.

\begin{figure}[t]
    \centering
    \includegraphics[width=0.32\textwidth]{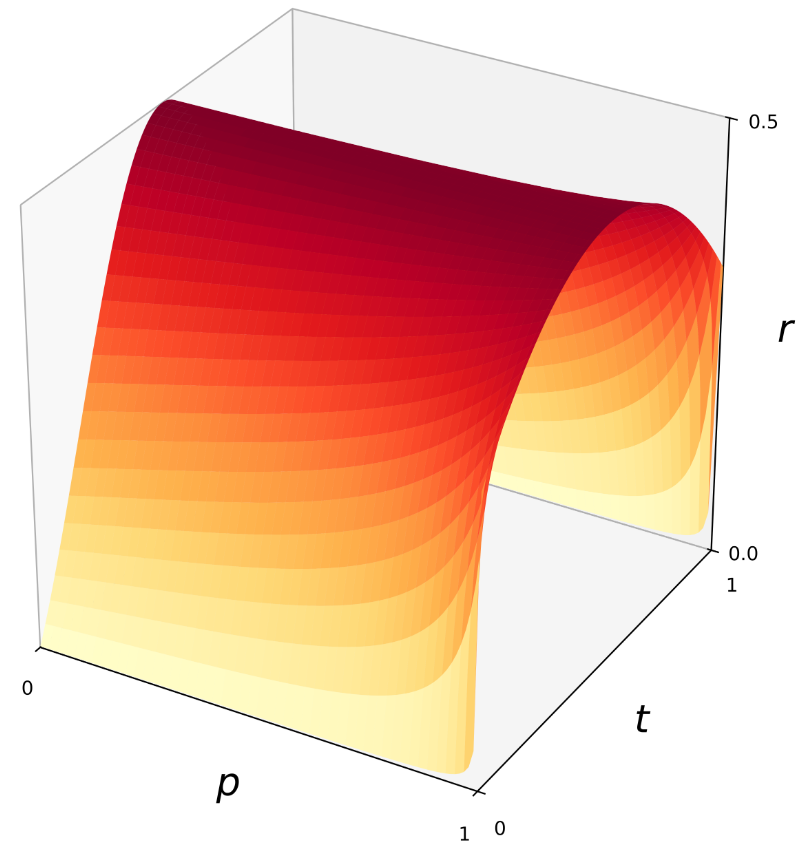}
    \caption{The stationary misprediction rate as a function of the defense parameter $p$ and branch taken probability $t$}
    \label{fig:misp}
\end{figure}


Since $m$ uniformly scales all non-self-loop transitions (cf.~Fig.~\ref{fig:2bitEnhancedPSCprobMooreMachines}), it affects convergence speed but not equilibrium: smaller $m$ only makes the PSC dwell longer in each state, while transition directions are governed by $t$ and $p$.
Thus the stationary distribution and misprediction rate are determined entirely by $p$ and $t$.
For $p = 0$, the rate reduces to that of the deterministic SC.
Fig.~\ref{fig:misp} shows $r$ for $p \in [0, 1]$ and $t \in [0.001, 0.999]$: for fixed $t \in (0,1)$, $r$ increases monotonically with $p$, reaching $0.5$ at $p = 1$; Sect.~\ref{ssec:privacy-utility} discusses the privacy--utility implication.


\subsection{Privacy-utility trade-off}
\label{ssec:privacy-utility}

The DP guarantee of the enhanced PSC (Thm.~\ref{thm:pTwoBitDifferentialPrivacyResult}) and its stationary misprediction rate (Sect.~\ref{ssec:counter-misprediction}) induce a privacy-utility curve.
Following the DP literature~\cite{DR:14:AFDP,GhoshRS2012,GengV2016}, we define utility as prediction accuracy $u = 1 - r$.

\paragraph*{The privacy-utility curve.}
By Thm.~\ref{thm:pTwoBitDifferentialPrivacyResult}, $\epsilon = \abs{\ln \frac{p}{1-p}}$, decreasing on $p \in (0,0.5]$ and increasing on $p \in [0.5,1)$.
For fixed $t \in (0,1)$, $r$ increases monotonically with $p$ (Sect.~\ref{ssec:counter-misprediction}).
Thus improving privacy requires moving $p$ toward $0.5$, which increases $r$.
At $p=0.5$ we obtain perfect privacy ($\epsilon=0$) but $r=0.5$; as $p$ tends to $0$ (or $1$), prediction approaches the deterministic SC while $\epsilon$ tends to $\infty$.

\paragraph*{Optimal defense parameter selection.}
For any target privacy budget $\epsilon > 0$, there exist two values $p_{1} = \frac{1}{1 + e^{\epsilon}} < \frac{1}{2}$ and $p_{2} = \frac{e^{\epsilon}}{1 + e^{\epsilon}} > \frac{1}{2}$ that achieve the same $(\epsilon, 0)$-differential privacy. 
Since $r$ is monotonically increasing in $p$, choosing $p_{1} < 0.5$ always yields a strictly lower misprediction rate than $p_{2} > 0.5$ for the same privacy guarantee.

\begin{proposition}
\label{prop:optimalP}
    For any target privacy budget $\epsilon > 0$ and branch probability $t \in (0, 1)$, the defense parameter $p^{*} = \frac{1}{1 + e^{\epsilon}}$ minimizes the misprediction rate among all values of $p$ that satisfy $(\epsilon, 0)$-differential privacy.
\end{proposition}
\begin{proof}
    By Thm.~\ref{thm:pTwoBitDifferentialPrivacyResult}, the set of $p$ values achieving exactly $(\epsilon, 0)$-DP is $\{p_{1}, p_{2}\}$ where $p_{1} = \frac{1}{1+e^{\epsilon}}$ and $p_{2} = \frac{e^{\epsilon}}{1+e^{\epsilon}} = 1 - p_{1}$. 
    Since $r$ is monotonically increasing in $p$ and $p_{1} < p_{2}$, we have $r(p_{1}, t) < r(p_{2}, t)$ for all $t \in (0,1)$.
    Moreover, any $p'$ satisfying $(\epsilon', 0)$-DP with $\epsilon' \leq \epsilon$ must have $p' \in [p_{1}, p_{2}]$, which gives $r(p', t) \geq r(p_{1}, t)$.
\end{proof}

Consequently, the defense parameter should always be chosen from the interval $(0, 0.5]$. 
For example, to achieve $(\ln 9, 0)$-DP, one should set $p^{*} = 0.1$ (not $p = 0.9$), which significantly reduces the misprediction rate while providing the same privacy guarantee.

\paragraph{Comparison with the randomized response mechanism.}
An alternative is the \emph{randomized response} (RR) mechanism~\cite{Warner1965,DR:14:AFDP}, which in our setting randomizes the branch input at \emph{every} counter state: the true direction $b$ is retained with probability $p_{\mathsf{RR}} = \frac{e^{\epsilon}}{1 + e^{\epsilon}}$ and flipped otherwise. 
By DP post-processing~\cite{DR:14:AFDP}, composing RR with any deterministic saturating counter yields an $(\epsilon,0)$-DP mechanism for the attack output.

The key difference is \emph{where} randomization is applied: RR perturbs every branch input, whereas our enhanced PSC randomizes only at the strong states ($\strongtaken$/$\strongnotaken$), preserving weak-state updates and yielding better accuracy for the same privacy guarantee.
Under RR, the SC perceives an effective taken probability $t_{\mathsf{eff}} = t \cdot p_{\mathsf{RR}} + (1-t)(1-p_{\mathsf{RR}})$; the resulting misprediction rate is:
\begin{equation}
\label{eq:rrMisprediction}
    r_{\mathsf{RR}} = \frac{(1{-}t) \cdot t_{\mathsf{eff}}^2(2{-}t_{\mathsf{eff}}) + t \cdot (1{-}t_{\mathsf{eff}})^2(1{+}t_{\mathsf{eff}})}{1 - t_{\mathsf{eff}}(1{-}t_{\mathsf{eff}})}.
\end{equation}

\begin{table}[t]
    \caption{Misprediction rates: enhanced PSC (with $p^{*} = \frac{1}{1+e^{\epsilon}}$) vs.\ randomized response (RR) on a deterministic 2-bit SC, for selected $\epsilon$ and branch probability $t$. 
    $r_{\mathrm{SC}}$ is the rate of the baseline deterministic SC ($p=0, m=1$), and Reduction is $(r_{\mathsf{RR}} - r_{\mathsf{PSC}}) / (r_{\mathsf{RR}} - r_{\mathrm{SC}})$.}
    \label{tab:rr-comparison}
    \centering
    \begin{tabular*}{\linewidth}{@{\extracolsep{\fill}}cccccc@{}}
        \toprule
        $\epsilon$ & $t$ & $r_{\mathrm{SC}}$ & $r_{\mathrm{PSC}}$ & $r_{\mathrm{RR}}$ & Reduction \\
        \midrule
        \multirow{3}{*}{$\ln 9 \approx 2.20$} 
        & 0.9 & 0.117 & 0.129 & 0.155 & 68\% \\
        & 0.8 & 0.251 & 0.265 & 0.287 & 63\% \\
        & 0.7 & 0.377 & 0.386 & 0.399 & 61\% \\
        \midrule
        \multirow{3}{*}{$\ln 3 \approx 1.10$} 
        & 0.9 & 0.117 & 0.147 & 0.255 & 78\% \\
        & 0.8 & 0.251 & 0.285 & 0.357 & 68\% \\
        & 0.7 & 0.377 & 0.399 & 0.435 & 62\% \\
        \bottomrule
    \end{tabular*}
\end{table}

Table~\ref{tab:rr-comparison} compares the two approaches for representative privacy budgets and branch probabilities. 
In all cases, the enhanced PSC achieves a lower misprediction rate than RR; the advantage is most pronounced for highly biased branches and tighter privacy budgets (up to $78\%$ reduction at $\epsilon = \ln 3$, $t = 0.9$).

\paragraph{Practical parameter guidelines.}
Given a target privacy budget $\epsilon$, one should set $p^{*} = \frac{1}{1+e^{\epsilon}}$ (Prop.~\ref{prop:optimalP}).
Since neither $\epsilon$ nor~$r$ depends on~$m$ (cf.~Thm.~\ref{thm:pTwoBitDifferentialPrivacyResult} and Sect.~\ref{ssec:counter-misprediction}), the update probability $m$ can be tuned independently for convergence speed; $m = 0.5$~\cite{JCST2021zhao} provides a balanced default.
For example, $\epsilon = \ln 9 \approx 2.20$ with $p = 0.1$ and $m = 0.5$ yields at most $1.2$ percentage points of additional misprediction (at $t = 0.9$) while bounding the adversary's success probability by $\frac{e^{\epsilon}}{1+e^{\epsilon}} = 0.9$.

\subsection{General \kbit PSCs}
\label{sec:general}

Our methodology generalizes to arbitrary \kbit designs; all formal definitions and proofs are in Appendix~\ref{appendix-formal}.

We first consider a standard \kbit PSC, obtained by extending the probabilistic transition function to a \kbit linear chain with $2^{k-1}-1$ $\weaktaken$ and $\weaknotaken$ states. 
Our analysis reveals that this naive extension remains vulnerable to the Prime+Probe attack. 
Specifically, by generalizing the transition matrix analysis of Sect.~\ref{ssec:PSCattackModel} to the \kbit case (see Appendix~\ref{appendix-formal} for the full proof), we can show that the comparison $\Pr[\attackout[\taken]=c] \bowtie \Pr[\attackout[\notaken]=c]$ reduces to $c \bowtie (2^{k-1}-1)/m$ for any comparison operator $\bowtie$. 
Hence, an attacker can distinguish branch outcomes with high confidence once the probe count $c$ significantly exceeds $(2^{k-1}-1)/m$, using the same threshold-based strategy as in the 2-bit case. 
This confirms that merely adding probabilistic updates without targeted defense transitions is insufficient for higher-order counters.

\begin{figure}[t]
\centering
    \resizebox{\linewidth}{!}{
    \begin{tikzpicture}
        \path[use as bounding box] (-2.1,1.6) rectangle (11.7,-3.6);
        \node[Takenode,fixedsize] (ST) at (0,0) {$\strongtaken$};
        \node[Takenode,fixedsize] (WT1) at ($(ST) + (3,0)$) {$\weaktaken_{1}$};
        \node[fixedsize] (WTd) at ($(WT1) + (3,0)$) {$\ldots$};
        \node[Takenode,fixedsize] (WTn) at ($(WTd) + (3,0)$) {$\weaktaken_{n{-}1}$};
        \node[Notakenode,fixedsize] (WNn) at ($(ST) + (0,-2)$) {$\weaknotaken_{n{-}1}$};
        \node[fixedsize] (WNd) at ($(WNn) + (3,0)$) {$\ldots$};
        \node[Notakenode,fixedsize] (WN1) at ($(WNd) + (3,0)$) {$\weaknotaken_{1}$};
        \node[Notakenode,fixedsize] (SN) at ($(WN1) + (3,0)$) {$\strongnotaken$};
                
        \draw (ST) to[bend left=20] node[above] {$\notaken/m(1{-}p)$} node[below] {$\taken/mp$} (WT1);
        \draw (ST) to[in=200,out=160,looseness=6] node[pos=1,anchor=north east] {$\taken/1{-}mp$} (ST);
        \draw (ST) to[in=110,out=70,looseness=5] node[above] {$\notaken/1{-}m(1{-}p)$} (ST);
        
        \draw (WT1) to[bend left=20] node[below] {$\taken/m$} (ST);
        \draw (WT1) to[bend left=20] node[above] {$\notaken/m$} (WTd);
        \draw (WT1) to[in=290,out=250,looseness=5] node[below] {$\taken/1{-}m$} (WT1);
        \draw (WT1) to[in=70,out=110,looseness=5] node[above] {$\notaken/1{-}m$} (WT1);

        \draw (WTd) to[bend left=20] node[below] {$\taken/m$} (WT1);
        \draw (WTd) to[bend left=20] node[above] {$\notaken/m$} (WTn);

        \draw (WTn) to[bend left=15] node[right] {$\notaken/m$} (SN);
        \draw (WTn) to[bend left=20] node[below] {$\taken/m$} (WTd);
        \draw (WTn) to[in=-20,out=20,looseness=6] node[pos=1,anchor=north west] {$\taken/1{-}m$} (WTn);
        \draw (WTn) to[in=70,out=110,looseness=5] node[above] {$\notaken/1{-}m$} (WTn);
        
        \draw (SN) to[bend left=20] node[below] {$\taken/m(1{-}p)$} node[above] {$\notaken/mp$}(WN1);
        \draw (SN) to[in=20,out=-20,looseness=6] node[pos=1,anchor=south west] {$\taken/1{-}m(1{-}p)$} (SN);
        \draw (SN) to[in=290,out=250,looseness=5] node[below] {$\notaken/1{-}mp$} (SN);

        \draw (WN1) to[bend left=20] node[above] {$\notaken/m$} (SN);
        \draw (WN1) to[bend left=20] node[below] {$\taken/m$} (WNd);
        \draw (WN1) to[in=290,out=250,looseness=5] node[below] {$\notaken/1{-}m$} (WN1);
        \draw (WN1) to[in=70,out=110,looseness=5] node[above] {$\taken/1{-}m$} (WN1);
        
        \draw (WNd) to[bend left=20] node[above] {$\notaken/m$} (WN1);
        \draw (WNd) to[bend left=20] node[below] {$\taken/m$} (WNn);

        \draw (WNn) to[bend left=12.5] node[left] {$\taken/m$} (ST);
        \draw (WNn) to[bend left=20] node[above] {$\notaken/m$} (WNd);
        \draw (WNn) to[in=160,out=200,looseness=6] node[pos=1,anchor=south east] {$\taken/1{-}m$} (WNn);
        \draw (WNn) to[in=250,out=290,looseness=5] node[below] {$\notaken/1{-}m$} (WNn);
    \end{tikzpicture}
    }
    \caption{The probabilistic Moore machine for the enhanced \kbit PSC; $n = 2^{k-1}$}
    \label{fig:enhancedKbitPSC}
\end{figure}
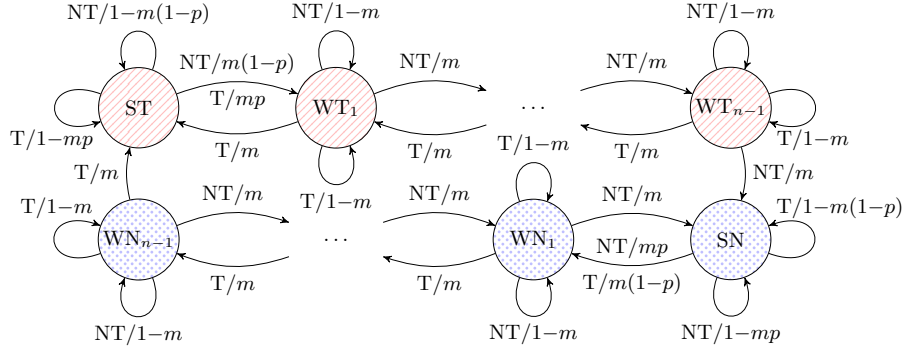

To address this, we propose the enhanced \kbit PSC, injecting the defense parameter $p$ exclusively into transitions from $\strongtaken$ and $\strongnotaken$ (Fig.~\ref{fig:enhancedKbitPSC}); intermediate states retain their original update $m$.
This generalized design, shown in Fig.~\ref{fig:enhancedKbitPSC}, preserves the rigorous security guarantees of the 2-bit version:
\begin{restatable}{theorem}{ppnbitDifferentialPrivacyResult}
\label{thm:ppnbitDifferentialPrivacyResult}
    Given an enhanced \kbit PSC and $p \in (0,1)$, if $p \geq \frac{1}{2}$, then it is $(\ln\frac{p}{1 - p}, 0)$-differentially private;
    otherwise, it is $(\ln\frac{1 - p}{p}, 0)$-differentially private.   
\end{restatable}

For the general enhanced \kbit PSC, we can prove that its misprediction rate is independent of the update probability $m$.
\begin{restatable}{theorem}{mispredictionRateIndependenceOfThreshold}
\label{thm:mispredictionRateIndependenceOfThreshold}
   Given an enhanced \kbit PSC, its misprediction rate is independent of $m$.
\end{restatable}

The proof (Appendix~\ref{appendix-formal}) relies on the decomposition $M = mB + (1-m)I$, where $B$ is independent of $m$; solving $\mu M = \mu$ then reduces to $\mu(B - I) = 0$, which is independent of $m$.

In summary, extending the PSC model to the \kbit case not only enhances its flexibility but also maintains its security and performance characteristics.
The fact that both the DP guarantee and the misprediction rate are independent of the counter width $k$ (in the sense that the same $\epsilon$-$p$ relationship holds) further underscores the generality of the proposed design.

\section{Experimental Evaluation}
\label{sec:experiments}

In this section, we evaluate the performance of PSCs under various parameters using real-world programs.
For comparison, we also analyze the experimental results of deterministic saturating counters and original PSCs. 
We utilize the clock-level Gem5 simulator~\cite{binkert2011GEM5} to simulate an out-of-order execution processor based on the latest Intel Sunny Cove core model~\cite{Wikichip}.
The PSCs model is implemented within the classic Tournament branch predictor architecture~\cite{kessler1999the}, with probabilistic mechanisms realized through Register Transfer Level code.
The formal model maps to the implementation as follows: each Moore machine state corresponds to a saturating counter value; $m$ and $p$ are realized by linear-feedback shift registers. Experiments validate \emph{utility} (prediction accuracy) under realistic workloads; privacy is established by the formal theorems of Sect.~\ref{sec:enhancedPSC}, not by empirical observation.
To assess performance, we employ the SPEC CPU 2017 benchmark suite~\cite{bucek2018Spec}. 
After warming up with 100 million instructions, the simulator executes further 100 million instructions in clock-precise mode. 
We measure \emph{mispredictions per kilo-instruction (MPKI)} to gauge branch prediction accuracy and \emph{instructions per cycle (IPC)} to evaluate overall performance.

\subsection{Performance evaluation}

\begin{figure}[t]
    \centering
    \includegraphics[width=1\linewidth]{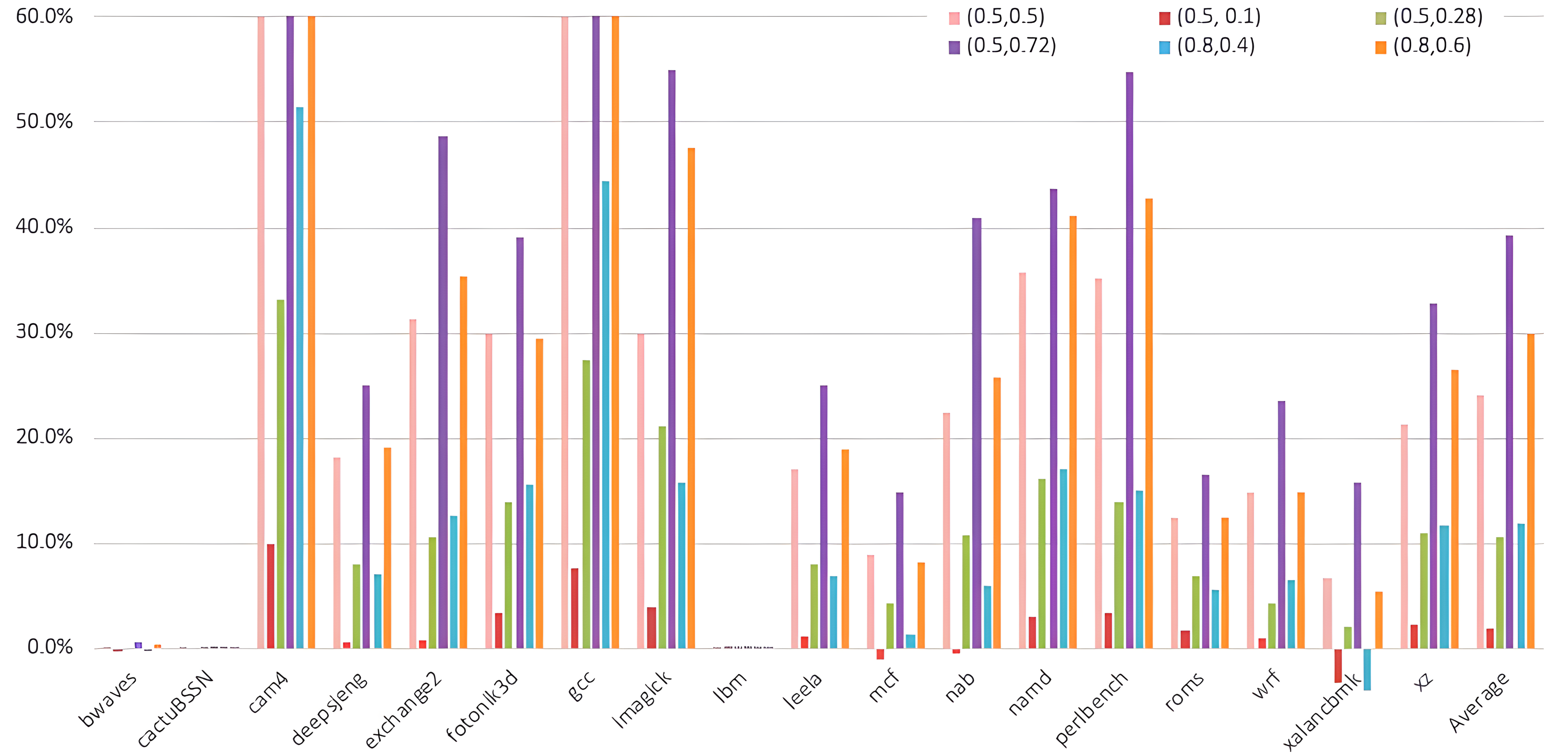}
    \caption{The normalised performance overhead for PSCs under different settings of parameters $(m,p)$, where the baseline is the deterministic saturating counter}
    \label{fig:overhead}
\end{figure}
Figure~\ref{fig:overhead} presents the performance evaluation results.
We measure Enhanced PSC overhead relative to the deterministic baseline across six configurations grouped by DP guarantee:
(i) $(m,p)=(0.5, 0.5)$, yielding perfect $(0, 0)$-DP;
(ii) $(0.5, 0.1)$, satisfying $(\ln{9}, 0)$-DP; and
(iii) $\{(0.5, 0.28), (0.5, 0.72), (0.8, 0.4), (0.8, 0.6)\}$, all providing $(0.94, 0)$-DP.
These budgets align with real-world deployments~\cite{tang2017privacylossapplesimplementation,10.1145/3183713.3196906,10.1145/2660267.2660348}.
The $(m, p) = (0.5, 0.1)$ setting gives a good balance, with $1.8\%$ average normalized overhead under $(\ln{9}, 0)$-DP over all benchmarks except the outlier \texttt{cam4}.
The perfect-privacy endpoint $(0.5, 0.5)$ incurs $24.1\%$ average overhead because it drives the stationary misprediction rate toward $50\%$; it marks the privacy-utility frontier rather than a recommended operating point.
Among the $(0.94,0)$-DP configurations, overhead increases with $p$, confirming that smaller $p$ with larger $m$ yields better performance.
Some stochastic configurations even outperform the deterministic baseline (e.g., a $4.1\%$ gain on \texttt{xalancbmk} for $(0.8,0.4)$), consistent with prior observations that randomness can prevent predictors from getting stuck in suboptimal states~\cite{JCST2021zhao}.

\subsection{Comparing theoretical and experimental misprediction rates}
\label{sec:experiment2}

\begin{table}[t]
    \caption{The experimental and theoretical misprediction rate on the MergeSort algorithm. PSC($m$,$p$) denotes the PSC with parameters $m$ and $p$.}
    \label{tab:prob}
    \centering
    \setlength{\tabcolsep}{6pt}
    \begin{tabular}{@{\extracolsep{4pt}}cclllllll@{}}
        \toprule
        \multirow{2}{*}{Data} &
        \multirow{2}{*}{Branch} &
        \multicolumn{1}{c}{\multirow{2}{*}{$t$}} &
        \multicolumn{2}{c}{SC} &
        \multicolumn{2}{c}{PSC(0.5, 0.5)} & \multicolumn{2}{c}{PSC(0.8, 0.4)} \\
        \cmidrule{4-5} \cmidrule{6-7} \cmidrule{8-9}
        & & & \multicolumn{1}{c}{$P_{\mathit{th}}$} & \multicolumn{1}{c}{$P_{\mathit{exp}}$} & 
        \multicolumn{1}{c}{$P_{\mathit{th}}$} & \multicolumn{1}{c}{$P_{\mathit{exp}}$} &
        \multicolumn{1}{c}{$P_{\mathit{th}}$} & \multicolumn{1}{c}{$P_{\mathit{exp}}$} \\
        \midrule
        \multirow{4}{*}{Uniform} 
        & 1st & 0.939 & 0.068 & 0.061 & 0.114 & 0.094 &  0.104 & 0.089 \\
        & 2nd & 0.495 & 0.500 & 0.510 & 0.500 & 0.504 &  0.500 & 0.506 \\
        & 3rd & 0.355 & 0.433 & 0.406 & 0.458 & 0.471 &  0.453 & 0.469 \\
        & 4th & 0.437 & 0.487 & 0.544 & 0.492 & 0.514 &  0.491 & 0.525 \\
        \midrule
        \multirow{4}{*}{Sorted} 
        & 1st & 0.891 & 0.128 & 0.109 & 0.194 & 0.154 & 0.179 & 0.149 \\
        & 2nd & 1 & 0 & 0 & 0 & 0 & 0 & 0 \\
        & 3rd & 0 & 0 & 0 & 0 & 0 & 0 & 0 \\
        & 4th & 0.895 & 0.123 & 0.105 & 0.188 & 0.158 & 0.173 & 0.154 \\
        \bottomrule
    \end{tabular}%
\end{table}

To validate the theoretical rates from Sect.~\ref{ssec:counter-misprediction}, we follow~\cite{elkhouly20152bit} and apply PSCs to MergeSort (Appendix~\ref{sec:mergesort}) on uniform and pre-sorted datasets of $10^5$ integers.
We compare $P_{\mathit{th}}$ with the observed $P_{\mathit{exp}}$ across four sensitive branches; Table~\ref{tab:prob} shows close agreement, validating both the stationary analysis and the assumption that this workload reaches stationarity.
Consistent with Sect.~\ref{ssec:counter-misprediction}, PSCs with $p > 0$ have higher misprediction rates than deterministic SCs when $t$ deviates from $0.5$, so $p$ should be minimized subject to the target DP budget.
For shorter executions, the stationary rate should be interpreted as an asymptotic long-run reference.

\section{Conclusion}   \label{sec:conclusion}
We presented a formal DP analysis of PSCs under Prime+Probe attacks by modeling counters and attacks as probabilistic Moore machines.
The analysis yields optimal attacks against existing PSCs and motivates an enhanced PSC whose parameter $p$ enforces pure $\epsilon$-DP while $m$ controls convergence speed.
We derived stationary misprediction rates, established the privacy-utility trade-off, and showed that targeted randomization improves utility over randomized response; experiments corroborate the stationary analysis.
Our formal guarantees hold for the PSC primitive under the Prime+Probe attack model.
Future work includes richer predictors such as TAGE, automated synthesis via PRISM~\cite{KNP:11:PRISM}, and DP composition under repeated attacks.

\subsubsection*{Acknowledgements.}
Work supported in part by 
National Key R\&D Program of China No.2025YFE0220300,
the Beijing Natural Science Foundation Project No.\@ IS26039,
the CAS Project for Young Scientists in Basic Research under grant No.\@ YSBR-040,
NSFC under grant No.\@ 61836005,
the CAS Pioneer Hundred Talents Program,
and
the ISCAS New Cultivation Project ISCAS-PYFX-202201.
\newline\protect\includegraphics[height=8pt]{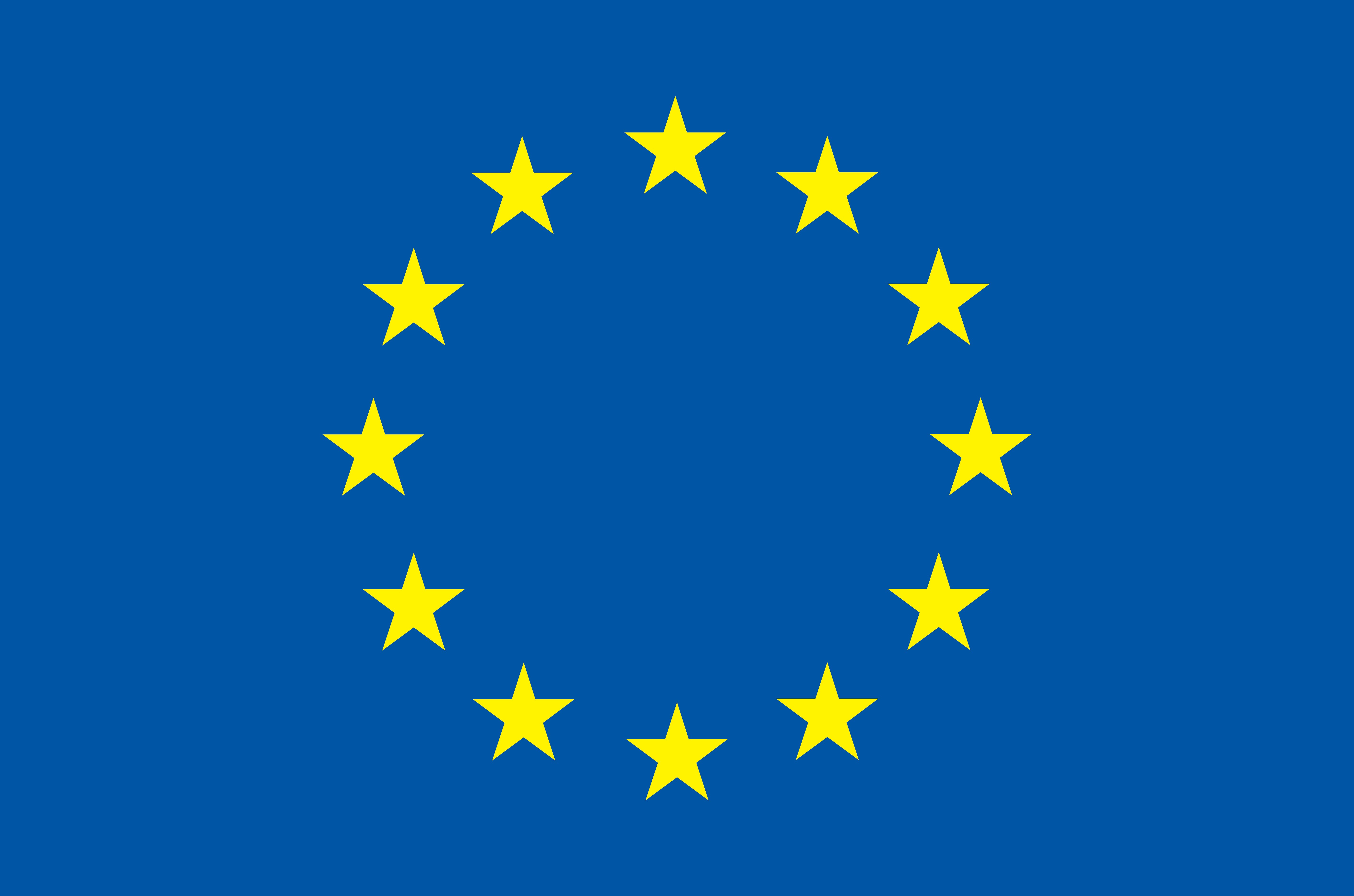} This project is part of the European Union's Horizon 2020 research and innovation programme under the Marie Sk\l{}odowska-Curie grant no.\@ 101008233.

\bibliographystyle{splncs04}
\bibliography{refs}

\clearpage
\appendix

\section{Calculation of the Misprediction Rate in the Stationary Distribution} 
\label{sec:rdetail}

Let $s = 1-t$ and $q = 1 - p$; 
we write the transition matrix $M$ of the PSC in Fig.~\ref{fig:enhancedPSC}, where the states are sorted in the order $\strongtaken$, $\weaktaken$, $\weaknotaken$, and $\strongnotaken$:
\[
\setlength{\arraycolsep}{5pt}
    M = \begin{bmatrix}
     1 - m + m(qt + ps) & m(qs + pt) & 0 & 0\\
     mt & 1 - m & 0 & ms\\
     mt & 0 & 1 - m & ms\\
     0 & 0 & m(qt + ps) & 1 - m + m(qs + pt)
    \end{bmatrix}
\]
When $m=0$, the transition matrix is the unit matrix, which implies that any distribution can be a stationary distribution, and in this case the saturating counter degenerates into a static branch predictor and stays in a single state forever, so we assume that $m \ne 0$. 
For programs that use saturating counters for branch prediction, the counters will converge to the stationary distribution after a large enough number of runs. 
Let $\mu = [a,b,c,d]$ denote the stationary distribution of $M$, where $a + b + c + d = 1$. 
The stationary distribution of $M$ can be computed by solving $\mu \cdot M = \mu$, whose solution is obtained as 
\begin{align*}
    \mu = \bigg [
    &\frac{t(qt + ps)}{t(qt + ps)(1 + qs + pt) + s(qs + pt)(1 + qt + ps)}, \\ 
    &\frac{t(qt + ps)(qs + pt)}{t(qt + ps)(1 + qs+pt) + s(qs + pt)(1 + qt + ps)}, \\ 
    &\frac{s(qt + ps)(qs + pt)}{t(qt + ps)(1 + qs + pt) + s(qs + pt)(1 + qt + ps)}, \\ 
    &\frac{s(qs + pt)}{t(qt + ps)(1 + qs + pt) + s(qs + pt)(1 + qt + ps)}
    \bigg].
\end{align*}
The misprediction rate $r$ can be obtained by calculating the sum of the probability of executing $\notaken$ in $\strongtaken$ and $\weaktaken$ and that of executing $\taken$ in $\strongnotaken$ and $\weaknotaken$, i.e., $r = (a + b) \cdot s + (c + d) \cdot t$. 
Then we have 
\[
    r = \frac{ts(qt + ps)(1 + qs + pt) + ts(qs + pt)(1 + qt + ps)}{t(qt + ps)(1 + qs + pt) + s(qs + pt)(1 + qt + ps)}.
\]

\section{The MergeSort Algorithm in Sect.~\ref{sec:experiment2}} \label{sec:mergesort}

The MergeSort algorithm used in Sect.~\ref{sec:experiment2} is shown as Alg.~\ref{algorithm:merge}. 
It is easy to see that there are four branches, namely, Line~\ref{branch1},~\ref{branch2},~\ref{branch3},~and~\ref{branch4} in Alg.~\ref{algorithm:merge}, that depend on the private information contained in the $\mathit{list}$ array.

\begin{algorithm}
  \caption{MergeSort}
  \label{algorithm:merge}
  \begin{algorithmic}[1]
  \Require{The integer array $\mathit{list}$, the left index $\mathit{low}$, the right index $\mathit{high}$}
  \Ensure{$\mathit{list}[\mathit{low}, \mathit{high})$ is sorted after running the algorithm}
    \Procedure{MergeSort}{$\mathit{list}, \mathit{low}, \mathit{high}$}
    \If{$\mathit{low} + 1 \geq \mathit{high}$}
    \Return
    \EndIf
    \State{$m \gets (\mathit{high} - \mathit{low})/2$}
    \State{\Call{MergeSort}{$\mathit{list}, \mathit{low}, \mathit{low} + m$}}
    \State{\Call{MergeSort}{$\mathit{list}, \mathit{low} + m, \mathit{high}$}}
    \State{$\mathit{leftList} \gets \mathit{list}[\mathit{low} : \mathit{low} + m]$}
    \Comment{Copy the list}
    \State{$\mathit{rightList} \gets \mathit{list}[\mathit{low} + m : \mathit{high}]$}
    \State{$i \gets 0$, $j \gets 0$, $k \gets 0$}
    \While {$i \neq \mathit{leftList}.\mathrm{size}()$ \& $j \neq \mathit{rightList}.\mathrm{size}()$}\label{branch1}
        \If {$\mathit{leftList}[i] \leq \mathit{rightList}[j]$}\label{branch2}
            \State{ $\mathit{list}[\mathit{low} + k] \gets \mathit{leftList}[i ++]$}
        \Else
            \State{$\mathit{list}[\mathit{low} + k] \gets \mathit{rightList}[j ++]$}
        \EndIf
        \State{$k \gets k+1$}
    \EndWhile
    \While {$i \neq \mathit{leftList}.\mathrm{size}()$}\label{branch3}
        \State{$\mathit{list}[\mathit{low} + k] \gets \mathit{leftList}[i ++]$}
        \State{$k \gets k+1$}
    \EndWhile
    \While {$j \neq \mathit{rightList}.\mathrm{size}()$}\label{branch4}
        \State{$\mathit{list}[\mathit{low} + k] \gets \mathit{rightList}[j ++]$}
        \State{$k \gets k+1$}
    \EndWhile
    \State \Return
    \EndProcedure
  \end{algorithmic}
\end{algorithm}

\section{Formal Modeling of \kbit Saturating Counters and Proofs}
\label{appendix-formal}
In this section, we provide the formal definition of the \kbit PSCs introduced in Sect.~\ref{sec:general}, along with the proofs of the theorems stated in the paper.

\begin{definition}
    Given $k \in \naturals_{\geq 2}$, let $n = 2^{k-1}$.
    A \kbit SC is a Moore machine $M = (S, \minit, \MooreTransition, \Outfunc)$, 
    where $S = \setnocond{\strongtaken, \strongnotaken} \cup \setcond{\weaktaken_{j}, \weaknotaken_{j}}{1 \leq j < n}$, $\MooreTransition$ is defined as 
    \begin{align*}
        &\MooreTransition(\strongnotaken, \notaken) = \strongnotaken, &&\MooreTransition(\strongnotaken, \taken) = \weaknotaken_{1}, \\
        &\MooreTransition(\weaknotaken_{1}, \notaken) = \strongnotaken, &&\MooreTransition(\weaknotaken_{i}, \taken) = \weaknotaken_{i+1}, \\
        &\MooreTransition(\weaknotaken_{i+1}, \notaken) = \weaknotaken_{i}, &&\MooreTransition(\weaknotaken_{n-1}, \taken) = \strongtaken, \\
        &\MooreTransition(\strongtaken, \notaken) = \weaktaken_{1}, &&\MooreTransition(\strongtaken, \taken) = \strongtaken, \\
        &\MooreTransition(\weaktaken_{i}, \notaken) = \weaktaken_{i+1}, &&\MooreTransition(\weaktaken_{i+1}, \taken) = \weaknotaken_{i}, \\
        &\MooreTransition(\weaktaken_{n-1}, \notaken) = \strongnotaken, &&\MooreTransition(\weaktaken_{1}, \taken) = \strongtaken,    \end{align*}
    where $1 \leq i < n-1$, and 
    $\Outfunc$ is defined as $\Outfunc(\strongtaken) = \Outfunc(\weaktaken_{j}) = \taken$ and $\Outfunc(\strongnotaken) = \Outfunc(\weaknotaken_{j}) = \notaken$ for each $1 \leq j < n$.
\end{definition}
A pictorial representation of a \kbit saturating counter is given in Fig.~\ref{fig:kbitSaturatingCounter}.

%
%
%
%

\begin{theorem} 
    The Prime+Probe attack on a \kbit saturating counter always succeeds, when Alg.~\ref{alg:cutoff} is called with input $n^{*} > 2^{k-1}$.
\end{theorem}
\begin{proof}
    Let $n^{*} \in \naturals$ satisfy $n^{*} > 2^{k-1}$.
    Initially, the saturating counter is in an unknown state $s$; 
    during Phase~1 of Alg.~\ref{alg:cutoff}, the transition $\MooreTransition(s, T)$ is repeated $n^{*}$ times.
    According to the definition of the \kbit saturating counter, we have that:
    \begin{itemize}
    \item
        a $\taken$ transition keeps the counter in $\strongtaken$;
    \item
        from any $\weaktaken_{i}$ state, it takes $i$ $\taken$ transitions to reach $\strongtaken$;
    \item
        from any $\weaknotaken_{i}$ state, it takes $n-i$ $\taken$ transitions to reach $\strongtaken$;
    \item
        from $\strongnotaken$, it takes $n$ $\taken$ transitions to reach $\strongtaken$, 
    \end{itemize}
    where $n = 2^{k-1}$ and $1 \leq i < n$.
    Thus, with $n^{*} > 2^{k-1} = n$, at the end of Phase~1 the current state of the saturating counter is $\strongtaken$.
    
    In Phase~2, the victim's thread executes its branch: 
    either it performs $v = \taken$ or $v = \notaken$.
    From this state, Phase~3 of Alg.~\ref{alg:cutoff} starts: 
    $\notaken$ is performed until the first correct prediction is obtained.
    
    If the victim's branch has direction $v = \taken$, then $\MooreTransition(\strongtaken, \taken) = \strongtaken$.
    From the $\strongtaken$ state, via $\notaken$ exactly $c = n$ branch executions are needed to reach the state $\strongnotaken$ where the first correct prediction occurs.
    
    If the victim's branch has direction $v = \notaken$, then $\MooreTransition(\strongtaken, \notaken) = \weaktaken_{1}$.
    From the $\weaktaken_{1}$ state, via $\notaken$ exactly $c = n - 1$ branch executions are needed to reach the state $\strongnotaken$ where the first correct prediction occurs.
    
    Since the number $k$ of bits in the saturating counter is known to the attacker, it can trivially derive the victim's branch direction $v$ from the value $c$ obtained from Alg.~\ref{alg:cutoff}: 
    if $c = 2^{k-1}$, then $v = \taken$, otherwise $v = \notaken$.
    Thus, the attack is always successful.
\end{proof}

\begin{definition}
\label{def:pkbit}
    Given $k \in \naturals_{\geq 2}$, let $n = 2^{k-1}$. A \kbit PSC is a 
probabilistic Moore machine $\mc = (\mstates, \minit, \mtransitions, \Outfunc)$, 
$S = \setnocond{\strongtaken, \strongnotaken} \cup \setcond{\weaktaken_{j}, \weaknotaken_{j}}{1 \leq j < n}$, $\mtransitions$ is defined as
\begin{align*}
    &\mtransitions(\strongnotaken, \notaken, \strongnotaken) = 1, &&\mtransitions(\strongnotaken, \taken, \weaknotaken_{1}) = m. \\
    &\mtransitions(\weaknotaken_{i}, \taken, \weaknotaken_{i+1}) = m, &&\mtransitions(\weaknotaken_{i+1}, \notaken, \weaknotaken_{i}) = m, \\
    &\mtransitions(\weaknotaken_{i}, \taken, \weaknotaken_{i}) = 1-m, &&\mtransitions(\weaknotaken_{i}, \notaken, \weaknotaken_{i}) = 1-m, \\
    &\mtransitions(\strongtaken, \taken, \strongtaken) = 1, &&\mtransitions(\strongtaken, \notaken, \weaktaken_{1}) = m, \\
    &\mtransitions(\weaktaken_{i}, \notaken, \weaktaken_{i+1}) = m, &&\mtransitions(\weaktaken_{i+1}, \taken, \weaktaken_{i}) = m, \\
    &\mtransitions(\weaktaken_{i}, \taken, \weaktaken_{i}) = 1-m, &&\mtransitions(\weaktaken_{i}, \notaken, \weaktaken_{i}) = 1-m, \\
    &\mtransitions(\weaktaken_{1}, \taken, \strongtaken) = m, &&\mtransitions(\weaktaken_{n-1}, \notaken, \strongnotaken) = m, \\
\end{align*}
where $m \in [0,1]$ and $1 \leq i < n-1$,  and 
    $\Outfunc$ is defined as $\Outfunc(\strongtaken) = \Outfunc(\weaktaken_{j}) = \taken$ and $\Outfunc(\strongnotaken) = \Outfunc(\weaknotaken_{j}) = \notaken$ for each $1 \leq j < n$.
\end{definition}


            


\begin{definition}\label{ineq:double1}
    For a \kbit PSC $\mc = (\mstates, \MooreInput, \mtransitions, \Outfunc)$, we define the random variable $\mathsf{out}$ as    
    \[
    \attackout = \min\{i \mid s_{i}=\strongnotaken\},
    \]
    where $v \in \{\taken, \notaken\}$, and $\strongtaken=s_{0},v,s_{1}, \notaken,s_{2}, \notaken, \dots, \notaken,s_{i+1}=\strongnotaken$ is a path of $\mc$. 
    A \pnbit prediction counter satisfies $(\epsilon, \delta)$-differential privacy, if for any $c \in \mathbb{N}^*$, we have
    \begin{equation}
  \begin{aligned}
  &\Pr[\attackout[\taken]=c] \leq \mathrm{e}^\epsilon \Pr[\attackout[\notaken]=c] + \delta, \\
  &\Pr[\attackout[\notaken]=c] \leq \mathrm{e}^\epsilon\Pr[\attackout[\taken]=c] + \delta ,
  \end{aligned}    
  \end{equation}
  where $\epsilon, \delta \in \posreals$.
\end{definition}    
For the sake of brevity in the subsequent discussion, we let $ \substituteg = 2^{k-1}-1$ in the proof.
\begin{restatable}{theorem}{attackStrategyUniparapnbit}
\label{thm:attackStrategyUniparapnbit}
    Given a \kbit PSC $\mc$, for any $c \in \naturals$ such that $c \geq 2^{k-1}-1$ and $\mathord{\bowtie} \in \setnocond{\mathord{<}, \mathord{\leq}, \mathord{\geq}, \mathord{>}, \mathord{=}}$,
    we have:
    \[
         \Pr[\attackout[\taken]=c] \bowtie \Pr[\attackout[\notaken]=c] \iff c \bowtie \frac{2^{k-1} -1}{\threshold}.
    \]
\end{restatable}
\begin{proof}

        For $\attackout[\taken] = c$, the corresponding path is $\mpath = s_{0}, \taken,s_{1}, \notaken,s_{2}, \notaken, \dots, \notaken$, $s_{c+1}$,
        where $s_{0} = \strongtaken, s_{c+1}=\strongnotaken$.
        From the definition of $\mtransitions$, there are only one probabilistic transition of it when accepting $\taken$, i.e., $\mtransitions(ST, \taken, ST) = 1$.
        so $s_{1} = \strongtaken$. 
        For $\forall i \in \setcond{l \in \naturals_{\geq1}}{ k \leq c}$, we have $\mtransitions(s_{i}, \notaken, s_{i+1}) = \threshold, \mtransitions(s_{i}, \notaken, s_{i}) = 1-\threshold$.
        Let $ \ranvare{1}, \dots, \ranvare{c}$ be the independent random variables, where
        \begin{align*}
            \ranvare{i} &= 
            \begin{cases}
                1, & s_{i} = s_{i+1} \\
                0, & \text{otherwise} \\
            \end{cases}
        \end{align*}

        From the definition of \kbit PSC, we know that $c \geq 2^{k-1}$, 
        and any path involving $\strongtaken$ and $\strongnotaken$ 
        contains $\substituteg$ distinct weak states taken, when $\attackout[T] = c,   
        \ranvars{c} = \sum_{i=1}^{c} \ranvare{i} = 2^{k-1}$. Then it follows that
        \[
            \Pr[\attackout[\taken]=c] = \Pr[\ranvars{c} = 2^{k-1}] = \binom{c}{2^{k-1}} \threshold^{2^{k-1}} (1-\threshold)^{c-2^{k-1}}.     
        \]

        For $\attackout[\notaken] = c$, the corresponding path is  $\mpath = s_{0}, \notaken,s_{1}, \notaken,s_{2}, \notaken, \dots, \notaken$, $s_{c+1}$,
        where $s_{0} = \strongtaken, s_{c+1}=\strongnotaken$.
        From the definition of $\mtransitions$, there are two probabilistic transitions of it when accepting $\notaken$, i.e., $\mtransitions(ST, \notaken, ST) = 1-\threshold$ and $ \mtransitions(ST, \notaken, \weaktaken_{1}) = \threshold$. 
        For $\forall i \in \setcond{k \in \naturals_{\geq 0}}{ k \leq c}$, we also have $\mtransitions(s_{i}, \notaken, s_{i+1}) = \threshold, \mtransitions(s_{i}, \notaken, s_{i}) = 1-\threshold$.
        
        \begin{itemize}
            \item If $s_{1} = \strongtaken$, similar to the case where $\attackout[\taken] = c$, we have
                \begin{align*}
                    \Pr[\attackout[\notaken]=c, s_{1} = \strongtaken] 
                     & = \mtransitions(\strongtaken, \notaken, \strongtaken) \cdot \Pr[\attackout[\taken]=c] \\
                     & = (1-\threshold) \cdot \binom{c}{2^{k-1}} \threshold^{2^{k-1}} (1-\threshold)^{c-2^{k-1}} \\
                \end{align*}
            \item If $s_{1} = \weaknotaken_{1}$, then for $\mpath = \strongtaken, \notaken, \weaktaken_{1}, \notaken,s_{2}, \notaken, \dots, \notaken,s_{c+1}$, its fragment $\mpath[2:]$ is a path starting from $\weaktaken_{1}$ and ending at $s_{c+1}$. 
            Let $ \ranvare{1}, \dots, \ranvare{c}$ be independent random variables,
            and since any path from $\weaktaken_{1}$ to $\strongnotaken$ contains $\substituteg$ distinct weak states,
            when $\attackout[\notaken] = c,   
            \ranvars{c} = \sum_{i=1}^{c} \ranvare{i} = \substituteg$, then it follows that
            \begin{align*}
                \Pr[\attackout[\notaken]=c, s_{1} = \weaktaken_{1}] 
                & = \mtransitions(\strongtaken, \notaken, \weaktaken_{1}) \cdot 
                    \Pr[\ranvars{c} = \substituteg]\\
                & = \threshold \cdot \binom{c}{\substituteg} \threshold^{\substituteg} (1-\threshold)^{c - \substituteg}
            \end{align*}
        \end{itemize}
        In summary, we have
        \[
            \Pr[\attackout[\notaken]=c] = \Pr[\attackout[\notaken]=c, s_{1} = \strongtaken] + \Pr[\attackout[\notaken]=c, s_{1} = \weaktaken_{1}]
        \]
        so $\Pr[\attackout[\taken]=c] \bowtie \Pr[\attackout[\notaken]=c]$ is equivalent to 
        \[
            \binom{c}{2^{k-1}} \threshold^{2^{k-1}} (1-\threshold)^{c-2^{k-1}} \bowtie \binom{c}{\substituteg} \threshold^{2^{k-1}-1} (1-\threshold)^{c - 2^{k-1}+1}
        \]
        Simplifying, we get:
        \[
            c \bowtie \frac{2^{k-1} -1}{\threshold}
            \qedhere
        \]
\end{proof}


\begin{corollary}
    Given a \kbit PSC $\mc$, for all $\epsilon >0$, $\mc$ does not satisfy $(\epsilon,0)$-differential privacy. 
\end{corollary}
\begin{proof}
    From the definition of \kbit PSC, we know that there is an unique
    path $\mpath$ of $\mc$ from $\strongtaken$ to $\strongnotaken$: $\mpath = \strongtaken, \notaken, \weaktaken_{1}, \dots, \weaktaken_{1}, \notaken, \strongnotaken$.
    other paths are longer than it.
    When $\attackout[\evaluation] = c = \substituteg$, the $\evaluation$ must be $\notaken$, 
    otherwise $c \geq 2^{k-1}$. This means $\Pr[\attackout[\taken]=c] = 0$ and $\Pr[\attackout[\notaken]=c] = 1$, 
    which implies $(\epsilon,0)$-differential privacy is not satisfied.
\end{proof}

\begin{definition}
\label{def:ppnbit}
Given $k \in \naturals_{\geq 2}$, let $n = 2^{k-1}$. The enhanced \kbit PSC is a 
probabilistic Moore machine $\mc = (\mstates, \minit, \mtransitions, \Outfunc)$, 
where $S = \{\strongtaken, \strongnotaken\} \cup \setcond{\weaktaken_{i}, \weaknotaken_{i}}{i \in [1,2^{k-1}-1], k\in \mathbb{N}}$, $\mtransitions$ is defined as
\begin{align*}
    &\mtransitions(\safepsnotaken, \taken, \safepsnotaken)=1-m + mp,&&\mtransitions(\safepsnotaken, \taken, \weaknotaken_{1})=m(1-p), \\
    &\mtransitions(\safepsnotaken, \notaken, \safepsnotaken)=1-m p,&&\mtransitions(\safepsnotaken, \notaken, \weaknotaken_{1})=mp, \\
    &\mtransitions(\weaknotaken_{i}, \taken, \weaknotaken_{i+1})=m,&&\mtransitions(\weaknotaken_{i+1}, \notaken, \weaknotaken_{i})=m, \\
    &\mtransitions(\weaknotaken_{i}, \taken, \weaknotaken_{i})=1-m,&&\mtransitions(\weaknotaken_{i}, \notaken, \weaknotaken_{i})=1-m, \\
    &\mtransitions(\safepstaken, \taken, \safepstaken)=1-mp,&&\mtransitions(\safepstaken, \taken, \weaktaken_{1})=mp, \\
    &\mtransitions(\safepstaken, \notaken, \safepstaken)=1-m + mp,&&\mtransitions(\safepstaken, \notaken, \weaktaken_{1})=m(1-p), \\
    &\mtransitions(\weaktaken_{i}, \notaken, \weaktaken_{i+1})=m,&&\mtransitions(\weaktaken_{i+1}, \taken, \weaktaken_{i})=m, \\
    &\mtransitions(\weaktaken_{i}, \taken, \weaktaken_{i})=1-m,&&\mtransitions(\weaktaken_{i}, \notaken, \weaktaken_{i})=1-m, \\
    &\mtransitions(\weaktaken_{1}, \taken, \strongtaken)=m,&&\mtransitions(\weaktaken_{n-1}, \notaken, \strongnotaken)=m, \\
\end{align*}
where $m,p \in (0,1), i \in \naturals^* $ and $1 \leq i < n-1$,
 and $\Outfunc$ is defined as $\Outfunc(\strongtaken) = \Outfunc(\weaktaken_{j}) = \taken$ and $\Outfunc(\strongnotaken) = \Outfunc(\weaknotaken_{j}) = \notaken$ for each $j \in [1, n-1]$.
\end{definition}

In the following, we will present the proof of Thm.~\ref{thm:pTwoBitDifferentialPrivacyResult}. When $k = 2$, this also serves as the proof of Theorem ~\ref{thm:pTwoBitDifferentialPrivacyResult}.
\pTwoBitDifferentialPrivacyResult*
\begin{proof}
    For $\attackout[\notaken] = c$, we have the path $\mpath = s_{0}, \notaken, s_{1}, \notaken, s_{2}, \notaken, \dots, \notaken, s_{c+1}$,
    where $s_{0} = \safepstaken$ and $s_{c+1} = \safepsnotaken$. 

    \begin{itemize}
        \item If $s_{1} = \safepstaken$, the fragment $\mpath[2:]$ of $\mpath$ is  $\safepstaken, \notaken, s_{2} , \notaken,s_3, \notaken, \dots, \notaken,s_{c+1}$.
        we can view the fragment $\mpath[2:]$ as a path corresponding to $\attackout[\notaken] = c - 1$ 
        with $\mc$, and the probability of this condition is:
        \begin{align*}
            \Pr[\attackout[\notaken]=c, s_{1} = \safepstaken]
              = {} &\mtransitions(\safepstaken, \notaken, \safepstaken) \cdot \Pr[\attackout[\notaken]=c-1]\\
             {} = {} &(1-\threshold+\threshold\privacy) \cdot \Pr[\attackout[\notaken]=c-1]
        \end{align*}
        
        \item If $s_{1} = \weaktaken_{1}$, then similarly to the proof in \ref{thm:attackStrategyUniparapnbit}, we have
        \[
            \Pr[\attackout[\notaken]=c, s_{1} = \weaktaken_{1}]  = \threshold(1-\privacy) \cdot \binom{c}{\substituteg} \threshold^{\substituteg} (1-\threshold)^{c - \substituteg}
        \]
    \end{itemize}
    In summary, we have such a recurrence relation:
    \begin{align*}
        & \hphantom{{}={}} \Pr[\attackout[\notaken]=c] \\
        & {}={} \Pr[\attackout[\notaken]=c, s_{1} = \safepstaken] + \Pr[\attackout[\notaken]=c, s_{1} = \weaktaken_{1}] \\
        & {}={} (1-\threshold+\threshold\privacy) \cdot \Pr[\attackout[\notaken]=c-1] + \threshold(1-\privacy) \cdot \binom{c}{\substituteg} \threshold^{\substituteg} (1-\threshold)^{c - \substituteg}
    \end{align*}
    
    For $\attackout[\taken] = c$, the associated path is $\mpath = s_{0}, \taken, s_{1}, \notaken,s_{2}, \notaken, \dots,$ $\notaken, s_{c+1}$,
    where $s_{0} = \safepsnotaken$ and $s_{c+1} = \safepstaken$.
    \begin{itemize}
        \item If $s_{1} = \safepsnotaken$, the fragment $\mpath[2:]$ of $\mpath$ is  $\safepsnotaken, \taken, s_{2} , \notaken,s_3, \notaken, \dots, \notaken,s_{c+1}$.
        we can also view this fragment as a path corresponding to $\attackout[\taken] = c - 1$ 
        of $\mc$, and the probability of this condition is:
        \begin{align*}
            \Pr[\attackout[\taken]=c, s_{1} = \safepsnotaken]
             & = \mtransitions(\safepstaken, \taken, \safepstaken) \cdot \Pr[\attackout[\notaken]=c-1]\\
             & = (1-\threshold\privacy) \cdot \Pr[\attackout[\notaken]=c-1]
        \end{align*}
        
        \item If $s_{1} = \weaknotaken_{1}$, then similarly to condition $\attackout[\notaken] = c$, we have
        \[
            \Pr[\attackout[\taken]=c, s_{1} = \weaknotaken_{1}]  = \threshold\privacy \cdot \binom{c}{\substituteg} \threshold^{\substituteg} (1-\threshold)^{c - \substituteg}
        \]
    \end{itemize}
    In summary, we also have a such recurrence relation:
    \begin{align*}
        \Pr[\attackout[\taken]=c]
        & = \Pr[\attackout[\taken]=c, s_{1} = \safepsnotaken] + \Pr[\attackout[\taken]=c, s_{1} = \weaknotaken_{1}] \\
        & = (1-\threshold\privacy) \cdot \Pr[\attackout[\notaken]=c-1] + \threshold\privacy \cdot \binom{c}{\substituteg} \threshold^{\substituteg} (1-\threshold)^{c - \substituteg}
    \end{align*}
    Let we denote $a_c = \binom{c}{\substituteg} \threshold^{\substituteg} (1-\threshold)^{c - \substituteg}$,
    $p_c = \Pr[\attackout[\notaken]=c]$ and $q_c = \Pr[\attackout[\taken]=c]$, then we have:
    \begin{align*}
        p_c & = \threshold(1-\privacy) \cdot a_c + (1-\threshold+\threshold\privacy) \cdot p_{c-1} \\
        q_c & = \threshold\privacy \cdot a_c + (1-\threshold\privacy) \cdot p_{c-1}
    \end{align*}
    
    From the definition \ref{ineq:double1}, if $\privacy \geq \frac{1}{2}$, we can see that $\mc$ satisfies $(\ln\frac{\privacy}{1- \privacy},0)$-differential privacy.
 Furthermore, if $\privacy < \frac{1}{2}$, we can also prove that $\mc$ satisfies $(\ln\frac{1-\privacy}{\privacy},0)$-differential privacy.
\end{proof}

\begin{definition}
    we have the following definition for the \privacypnbit prediction counter $\mc$:
    \begin{itemize}
        \item A random variable $\execute \colon \mstates \to \setnocond{\taken, \notaken}$, 
        $\forall s \in  \mstates$, it follows that 
        \[
            \Pr[\execute(s) = \taken] + \Pr[\execute(s) = \notaken] = 1
        \]
        \item If we sort the states in $\mstates$ as
        $S' = (\safepstaken, \weaktaken_{1} \dots, \weaktaken_{\substituteg}, \weaknotaken_{\substituteg}, \dots \weaknotaken_{1}, \safepsnotaken)$, 
        the \emph{transition matrix} of $\mc$, denoted as $\tranmatrix$, is a square matrix of order $\abs{\mstates}$,
        where $\tranmatrix_{i,j} = \sum_{v \in \setnocond{\taken, \notaken}}\mtransitions(S'[i],v,S'[j]) \cdot \Pr[\execute(S'[i]) = \evaluation]$ and $S'[i],S'[j]$ are the entries $i$-th and $j$-th in $S'$.
        \item The \emph{stationary probability} of state $s \in \mstates$ is denoted by
 $\pstationary_{s} = \sum_{i\in S} \pstationary_{i} \cdot \tranmatrix_{i,s}$.    
        \item The \emph{misprediction rate} of is denoted as $\sum_{s \in \mstates} (1-\Pr[\execute(s) = \Outfunc(s)]) * \pstationary_{s}$.
    \end{itemize}
    
\end{definition}

\mispredictionRateIndependenceOfThreshold*
\begin{proof}
    Firstly, we denote the $s = 1-t,q = 1-p,h = (qt + ps), l = (qs + pt)$ and $n = 1-m$, then we have the transition matrix $\tranmatrix$ of $\mc$ can be written as:
\[
       \tranmatrix= \begin{bmatrix}
            mh + n  & ml & 0      & 0        & \dots  & 0      & 0      & 0      & \dots  &  0        \\
            mt              & n          & ms     & 0        & \dots  & 0      & 0      & 0      & \dots  &  0         \\
            0               & mt         & n      & ms      & \dots  & 0      & 0      & 0      & \dots  &  0        \\
            \vdots          & \vdots     & \vdots & \vdots   & \ddots & \vdots & \vdots & \vdots & \vdots &\vdots  \\
            0               & 0          & 0      & 0        & \dots  & mt     & n      & 0      & \dots  & ms     \\
            mt              & 0          & 0      & 0        & \dots  & 0      & 0      & n      & \dots  & 0  \\
            0               & 0          & 0      & 0        & \dots  & 0      & 0      & mt     & \dots  & 0  \\
            \vdots          & \vdots     & \vdots & \vdots   & \ddots & \vdots & \vdots & \vdots & \vdots &\vdots\\
            0               & 0          & 0      & 0        & \dots  & 0      & 0      & mt     & n      & ms    \\
            0               & 0          & 0      & 0        & \dots  & 0      & 0      & 0      & mh  & ml + n\\
      \end{bmatrix}
\]
For the given matrix $M$, observations indicate that $n$ is always located on the main diagonal of the matrix, 
while the other off-diagonal elements are integer multiples of $m$. Based on this characteristic, the matrix $M$ can be expressed 
as the sum of two matrices, that is,
$M = mB + nI$,
where $B$ is a matrix independent of $m$, and $I$ denotes the identity matrix. 
Within this framework, for a probabilistic automaton with a finite state space and $M$ as its transition matrix, 
given that all states within the chain are mutually accessible, the chain exhibits irreducibility. Furthermore, due to the non-periodicity 
of the $\mc$, it ensures the existence of a unique stationary distribution $\pstationary$, and starting from any initial probability distribution, 
the system will ultimately converge to this stationary distribution over time.

Consequently, the linear equation for the stationary distribution can be formulated as $\pstationary = \pstationary(mB + nI)$, 
which has a unique solution $\pstationary$. By setting $n = 1 - m$ and substituting it into the aforementioned equation, 
we can derive the conclusion that $\pstationary(B - I) = 0$. It should be noted that since the matrix $B$ 
and the identity matrix $I$ are both independent of the parameter $m$, this implies that the stationary distribution 
$\pstationary$ is also independent of the value of $m$.
Thus, we can conclude that the misprediction rate of $\mc$ is independent of the threshold parameter $\threshold$.
\end{proof}
\end{document}